\documentclass[a4paper,11pt]{article}

\usepackage[utf8]{inputenc}
\usepackage[T1]{fontenc}
\usepackage{lmodern}
\usepackage{microtype}
\usepackage{lineno}

\usepackage[a4paper,top=1.0in,bottom=1.0in,left=1.0in,right=1.0in]{geometry}

\usepackage{amsmath,amssymb,amsthm,mathtools,bm} 

\newtheorem{theorem}{Theorem}[section]
\newtheorem{lemma}[theorem]{Lemma}
\newtheorem{corollary}[theorem]{Corollary}

\newtheorem{observation}[theorem]{Observation}
\newtheorem{definition}[theorem]{Definition}
\newtheorem{example}[theorem]{Example}
\newtheorem{claim}[theorem]{Claim}

\DeclareMathOperator{\OPT}{\sf OPT}
\DeclareMathOperator{\Alg}{\sf ALG}

\usepackage{algorithm,algorithmic}
\usepackage{graphicx}
\usepackage{booktabs}
\usepackage[font=small,labelfont=bf]{caption}
\usepackage{subcaption}

\usepackage[dvipsnames]{xcolor}
\usepackage{hyperref}
\hypersetup{
    colorlinks=true,
    linkcolor=blue,
    citecolor=Green,
    urlcolor=MidnightBlue
}

\title{Online Contract Design}
\date{}

\author{
  Elad Lavi\thanks{Computer Science Department, Technion, Haifa, Israel. \texttt{elad.lavi@cs.technion.ac.il}} 
  \and
  Hadas Shachnai\thanks{Computer Science Department, Technion, Haifa, Israel. \texttt{hadas@cs.technion.ac.il}} 
  \and
  Inbal Talgam-Cohen\thanks{School of Computer Science, Tel Aviv University. \texttt{inbaltalgam@gmail.com}}
}

\begin{document}

\maketitle

\begin{abstract}
    We initiate the study of online contracts, which integrate the game-theoretic considerations of \emph{economic contract theory}, with the algorithmic and informational challenges of \emph{online algorithm design}. Our starting point is the classic online setting with preemption of Buchbinder \emph{et al.}~[SODA'15], in which a hiring principal faces a sequence of adversarial agent arrivals. Upon arrival, the principal must decide whether to tentatively accept the agent to their team, and whether to dismiss previous tentative choices. Dismissal is irrevocable, giving the setting its online decision-making flavor. In our setting, the agents are rational players: once the team is finalized, a game is played where the principal offers contracts (performance-based payment schemes), and each agent decides whether or not to work. Each working agent adds value to the principal's reward, and the goal is to choose a team that maximizes the principal's utility. Our main positive result is a $1/2$-competitive algorithm,
    which matches the best-possible competitive ratio. Our algorithm is randomized and this is necessary, as we show that no deterministic algorithm can attain a constant competitive ratio. Moreover, if agent rewards are allowed to exhibit combinatorial structure known as XOS, even randomized algorithms might fail. En route to our competitive algorithm, we develop the technique of \emph{balance points}, which can be useful for further exploration of online contracts in the adversarial model. 
\end{abstract}

\section{Introduction}

Contract theory studies how a \emph{principal} can incentivize agents to work, when work is costly for the agents and its outcome rewards the principal~\cite{holmstrom1979moral,holmstrom1982moral}. The economic tool for aligning incentives is \emph{contracts} --- performance-based payments schemes which the principal commits to. 
The theory of contracts is a cornerstone of microeconomics, recognized by the 2016 Nobel prize.
The rising research area of \emph{algorithmic} contract theory studies contracts through the computational lens. 
Early papers include~\cite{BabaioffFNW12,dutting2019simple}; for a survey of the recent flurry of research activity see~\cite{dutting2024algorithmic}. 

\emph{Multi-agent} contract settings are of particular computational interest~\cite{BabaioffFNW12,dutting2023multi}. In such settings, the principal chooses a subset of agents to work for them as a team; the agents' contracts must incentivize them to work, by disincentivizing them from free-riding on the efforts of other team members. Finding the ``right'' subset of agents to form a team is inherently a combinatorial problem, and is tackled by~\cite{BabaioffFNW12,dutting2023multi,VuongDPP23,CastiglioniEtAl23,CacciamaniEtAl24,dutting2025multi,TalCCELT25}, among others.
To date, the substantial body of work on multi-agent contracts has focused on the \emph{offline} setting, in which the set of agents is known in advance, so that the principal can optimize globally.  

Bringing algorithmic contract theory closer to real-world applications calls for models that capture how interactions unfold over time. In many classic and modern environments, agents arrive sequentially and unpredictably. For example, a team leader occasionally receives applications from potential new members, who may change the composition of the team; online labor platforms have a dynamic and fluctuating supply of workers; even when the team is composed of AI rather than human agents (see e.g.~\cite{AutoGen}), new AI models appear periodically. 
In this work we study the question: How should team contracts be designed when the potential participants are revealed gradually over time?

\emph{Online algorithms} provide a classic framework and rich body of techniques for decision-making in settings characterized by sequential and uncertain arrivals~\cite{BorodinE05}. In this paper, we marry this framework with contracts, and initiate the study of \emph{online contract design}. 

\subsection{Our Contribution: OMAC Model}

As a first step towards integrating the two fields, we formulate the \emph{Online Multi-Agent Contract (OMAC)} model. 
The OMAC model consists of sequential and uncertain agent arrivals that can conclude abruptly. We briefly describe our model, deferring details to Section~\ref{sec:preliminaries}.

\paragraph{Modeling Uncertainty.}
Our first design choice is how to approach uncertainty, where two standard approaches are the adversarial (worst-case) approach, and the Bayesian (stochastic) one.
We follow the classic work of Sleator and Tarjan~\cite{SleatorT85}, which introduces the \emph{competitive ratio under adversarial inputs} as the gold standard for online algorithms. The competitive ratio measures the relative performance of an online algorithm compared to an optimal offline one (see Section~\ref{sec:preliminaries} for a formal definition). 
The adversarial approach is particularly well-suited for situations with little to no data about the uncertain components of the setting,%
\footnote{For example, rapid technological change can invalidate previously learned distributional information. In light of the technological leaps in AI, there is much speculation but little to no data on the next generation of AI agents that will arrive on the market. Similarly, it is unknown how the worker population will evolve in the GenAI era, and this may render prior distributional knowledge obsolete.}  
and provides the canonical worst-case baseline for contracts under online arrivals.
Stochastic and stochastic-adversarial variants are also of interest
(see e.g.~\cite{samuelcahn1984comparison}), and there is forthcoming work on Bayesian online contracts~\cite{DuttingFG26}. 

\paragraph{A Combination of Settings: Online.}
We obtain our model by combining two fundamental settings: the adversarial online setting with preemption (e.g., \cite{BuchbinderFS15a}), and the multi-agent contract setting~\cite{dutting2023multi}. We now describe how these seemingly unrelated settings blend together naturally.

Buchbinder et al.~\cite{BuchbinderFS15a,BuchbinderFS19} describe an online hiring setting with arbitrarily-arriving agents.%
\footnote{They use different language, referring to candidate players hired by a sports team manager.} 
The team of agents~$S$ must maintain proficiency (as measured by a value function $f(S)$) at all times, since it is unknown when the arrivals will cease, or when the team will need to take action. For example, this is the case if the team's job is to respond to an emergency, recover from a system crash, or supply a surge in demand \cite{valentine2025flash}.

A central feature of Buchbinder et al.'s model is preemption --- the principal is allowed to replace an agent who was tentatively accepted as more promising agents arrive. While agent acceptance is tentative, agent dismissal is final and irrevocable, which is the main characteristic of online decision-making settings. 
Buchbinder et al.~show that preemption is key to circumventing intractability and achieving constant-factor guarantees, and there are many other such examples in the literature (see, e.g.~\cite{feldman2009online,RemovableKS}, and the multiple additional preemption related works surveyed in Section~\ref{sub:related}). Our model displays the same dichotomy: preemption is what allows meaningful guarantees to emerge, and without it we show an impossibility result (Theorem~\ref{theo: general negative}). 

\paragraph{A Combination of Settings: Contracts.}
In our model, after the online stage consisting of adversarial agent arrivals, tentative accepts, replacements, and irrevocable dismissals, the composition of the team is finalized. Since we model agents as rational and utility-maximizing, a game is then played: the principal designs and offers contracts to the team members, who in turn decide whether to work or not. A working agent incurs a cost and contributes additively to the principal's reward --- an additive instantiation of the multi-agent (offline) game formulated by D\"utting et al.~\cite{dutting2023multi}. 

As in the formulation of \cite{dutting2023multi}, the principal offers the agents \emph{linear} contracts, which are a fixed share of the reward eventually generated by the team. The fixed shares should sum up to at most~$1$ (the entire reward), or to less than $1$ in order to reward the principal. Intuitively, this directly connects the setting to the knapsack problem. For a formal connection to knapsack (a.k.a.~budgeted maximization), see Appendix~\ref{app:ks to omac}. 

The linear contracts can be designed according to a simple elegant formula of Babaioff et al.~\cite{BabaioffFNW12}, which induces an equilibrium in which all agents on the team are incentivized to work, while minimizing the principal's overall payment (equivalently, maximizing the principal's total utility). 
Thus, the main algorithmic challenge in OMAC is to maintain a competitive team during the online stage, while taking into account the overall contractual payment once the team is finalized.

\subsection{Our Contribution: Results and Techniques}

In Section~\ref{sec:rand-alg-additive} we design a randomized algorithm that achieves a competitive ratio of \(1/2\) for the OMAC model.
We complement our positive result by showing that no deterministic online algorithm achieves a constant competitive ratio. We also establish that randomization allows a competitive ratio of at most \(1/2\). The impossibilities are established through a direct adversarial construction (see Section~\ref{sec:negative}). 
Thus, our randomized algorithm achieves a tight, best-possible guarantee, and the case of additive rewards is fully resolved. 

Our tight positive result is obtained via a randomized strategy built around \emph{balance points} (see Section~\ref{sec:balance-points}). We introduce the structural concept of balance points to synchronize the agents' shares with the principal's utility as arrivals unfold. This technique departs sharply from prior online algorithms, where objectives are always nonnegative, and are often additive or submodular in the chosen set. Because we combine online algorithms with contract design, we must optimize the principal's incentive-induced utility, which can be negative and non-submodular even when the reward function has additive structure. 
As a result, existing online algorithms for settings with preemption fail in our model, as they might output solutions with arbitrarily negative utility. Our approach is explicitly designed to prevent such failures.

To place our approach in context, in Appendix~\ref{app:ks to omac} we formalize the connection between our problem and \emph{online budgeted maximization (knapsack) with preemption}, 
and show how a known randomized algorithm of~\cite{RemovableKS} for the latter problem can be adapted to OMAC. However, this algorithm is based on static thresholds, and achieves at most a $1/4$-competitive ratio for OMAC (compared to $1/2$ for online knapsack). In our approach, balance points serve as dynamic thresholds, and this flexibility is essential for attaining the optimal competitive ratio of $1/2$ for OMAC. Balance points thus supply the novel ingredient that optimally resolves the OMAC challenge. 

The connection to online budgeted maximization also reveals a barrier to going beyond additive, to submodular rewards: the problem of online monotone \emph{submodular} maximization subject to a budget constraint remains open in the literature (see~\cite{rawitz2021online}). We thus expect that solving OMAC with submodular rewards will be require new ideas and approaches. 

Beyond submodular, there are even more barriers: we can construct instances with XOS rewards (a superclass of submodular), where even randomized strategies fail to secure any constant-factor guarantee --- thus ruling out the possibility of meaningful performance in this richer domain (see Section~\ref{sec:negative}). 
Our construction exploits hidden complementarities among the agents. These cause an agent's value and share to crucially depend on the composition of the team, in a way that is inherently unpredictable in an online setting.

Together, our results show a range of online contract design problems, from manageable to intractable. 
Our resolution of the additive and XOS cases is a first step towards defining the boundary between tractable and hard reward functions in our model of online contracts.

\subsection{Related Work}
\label{sub:related}

\paragraph{Dynamic Contracts.} Temporal models of contracting have previously been explored by multiple works, all different from ours:
Several works on dynamic contracts have considered settings
where a single agent is engaged with repeatedly. E.g., in \cite{ezra2024contract}, the agent chooses actions sequentially, and in \cite{guruganesh2024contracting} the principal optimizes the contract against a no-regret learning agent.
Other works have focused on learning contracts over time, by interacting over many iterations with agents drawn from the same distribution, or alternatively with the same single agent. The principal's goal in the learning process is to adapt the incentives based on past outcomes and thus achieve low regret \cite{ho2014adaptive,ZhuEtAl22,BacchiocchiC0024,ChenCDH24,DuettingGSW23,GuruganeshSW23}. 
In \cite{collina2024repeated}, the principal induces a multi-round game by adaptively choosing the agent to contract with among $k$ candidates, again with the goal of no regret. 
Reinforcement learning approaches have also been considered \cite{ivanov2024principal,wu2024contractual,Bollini24}. Complementing this work, in concurrent forthcoming work by D\"utting et al.~\cite{DuttingFG26}, they consider Bayesian (rather than adversarial) online contract settings and achieve positive results as well as lower bounds.

\paragraph{Preemption and Consistency.}
Online maximization problems with preemption have been studied extensively and under many different names like \emph{free dismissal}, \emph{free disposal}, \emph{removable}, or \emph{recourse}.  
A closely related line of work develops \emph{consistent} algorithms, which bound the recourse under dynamic updates in problems like clustering and covering (e.g., \cite{Lattanzi2017,Fichtenberger2021,Gupta2020,Gupta2022,Lacki2021,Bhattachary2024}). In the terminology of D\"utting et al.~\cite{DuttingFLNSZ25}, our model is $1$-consistent, since each new arrival may cause at most one new hire (dismissed agents are never rehired). 
However, our algorithmic techniques and analysis are quite different as they are closely tied to the challenges of the contracting model. 
Closer to our model are works on constrained online submodular maximization. As discussed above, budget constraints are still open; the recent work of \cite{Yang024} studies cardinality constraints. 

\section{The OMAC Model}
\label{sec:preliminaries}

We introduce an online contract model, which can be viewed from two perspectives: First, as the \emph{online} version of the standard multi-agent contract setting~\cite{BabaioffFNW12,dutting2023multi}; second, as the \emph{strategic} version of the classic online decision problem with preemption~\cite{BuchbinderFS15a,ChanHJ+18}. 
We present our model beginning with its online aspects, then its strategic ones. 

\paragraph{Online Multi-Agent Contract (OMAC) Setting.}  

In the OMAC setting, a \emph{principal} interacts with an unknown set of \emph{agents} $N=[n]$, revealed sequentially in an arbitrary and unpredictable manner. The principal aims to hire a \emph{team}---a subset $S\subseteq N$ of agents. The team, once finalized, will be tasked with a project. Crucially, the principal does not know the total number of agents $n$ in advance. I.e., from the principal's point of view, the process may terminate after any iteration. The principal must ensure that at every time step, the quality of the current team is competitive. 

W.l.o.g., we call the $i$th agent to arrive agent $i$, and let $N_i=[i]$ be the first $i$ agents. At every step $i$, the principal chooses a tentative team $S_i$, subject to the restriction \(S_i \subseteq S_{i-1}\cup \{i\}\) (with the convention that $S_0=\emptyset$). Note that the team after step $i$ may or may not include the new agent $i$, and that previously chosen agents in $S_{i-1}$ can be dismissed. However, any agent \emph{not} tentatively chosen (i.e., any agent in $N_{i-1}\setminus S_{i-1}$) cannot be included in $S_i$. Thus, the principal cannot rehire dismissed agents. 

The team $S_n$ chosen in step~$n$ is final. I.e., after step $n$ is completed, the solution $S$ is set to $S_n$, and no additional dismissals are allowed. 

The above free dismissal assumption is precisely preemption, studied in many previous works on adversarial online algorithms. Since $|S_i\setminus S_{i-1}|\le 1$, it is also $1$-\emph{consistency}, studied e.g.~in \cite{DuttingFLNSZ25} and characteristic of the rapidly growing literature on \emph{stable} online algorithms (see Section~\ref{sub:related}).
Without free dismissal, a hired agent would necessarily remain on the team throughout the arrival of subsequent agents. We do not study this variant since, as shown in Section~\ref{app:dismiss-free}, it becomes impossible to guarantee any meaningful performance for the principal through an online strategy. 

\paragraph{The Principal's Reward Function.}  
The final team $S$ is delegated a project, whose outcome depends on the team's ensemble. As usual in multi-agent contract settings, this is captured by the \emph{reward function} $f: 2^N \rightarrow \mathbb{R}_{\geq 0}$, which is a set function  specifying the principal’s reward $f(S)$ from the project when carried out by team~$S$. As standard in the literature, we assume that $f$ is normalized ($f(\emptyset)=0$), and monotone ($f(T)\leq f(S)$ whenever $T\subseteq S\subseteq N$). This reflects that a project yields no reward given an empty team, and that extending a team can only increase the principal's reward. 
Our focus is on \emph{additive} set functions $f$, where $f(S) = \sum_{i \in S} f(\{i\})$ for all $S \subseteq N$, meaning each agent's contribution is independent of other agents' contributions.  
In Section~\ref{sec:negative} we also consider XOS set functions $f$, for which there exists a family of additive clauses $\{w^{(1)},\dots,w^{(m)}\}$ such that 
$f(S) = \max_{\ell \in [m]} w^{(\ell)}(S)$ for all $S \subseteq N$. In words, $f$ evaluates each team by the most favorable additive clause.

\paragraph{Agent Costs and Linear Contracts (Shares).}  

Each agent $i$ has a known \emph{cost} $c_i \in \mathbb{R}_{\geq 0}$ for working on the project. Thus, for $i$ to work as part of team $S$, the principal must pay $i$. In multi-agent contract settings, this is done by \emph{a priori} allocating to $i$ (before the project starts) an $\alpha_i$-fraction of the eventual reward from the project. This fraction $\alpha_i$ is known as a \emph{linear contract}~\cite{holmstrom1987aggregation}, and is also referred to as $i$'s \emph{share}~\cite{Aharoni0T25}. We write \(\alpha(S) = \sum_{i \in S} \alpha_i\) for the total share allocated to team \(S\), with the convention that \(\alpha(\emptyset) = 0\). 

A defining property of multi-agent contract settings is that once an agent is allocated a fraction of the future reward, this agent might be tempted to \emph{shirk}, allowing other agents to do the costly work while enjoying a fraction of the reward they generate. This is known in contract design as \emph{moral hazard}.

The agents' shares $(\alpha_i)_{i\in S}$ are designed to overcome moral hazard by making work an equilibrium, while simultaneously ensuring agents are well-compensated and protected in this dynamic setting, mirroring the robust participation incentives of our offline foundation~\cite{BabaioffFNW12}.
When this holds, we say that contract ${\bm \alpha}=(\alpha_i)_{i\in S}$ \emph{incentivizes} team $S$. There is a well-known closed formula for shares that incentivize team $S$, while maximizing the principal's remaining fraction of the project~\cite{BabaioffFNW12}:
\begin{equation}
    \label{eq: def alpha}
    \alpha_i := \frac{c_i}{f(S) - f(S \setminus \{i\})}.
\end{equation}
For every agent $i$, the fraction $\alpha_i$ increases as $i$'s cost $c_i$ increases, and somewhat counterintuitively, increases as $i$'s marginal contribution $f(S) - f(S \setminus \{i\})$ \emph{decreases} (when $i$'s work matters less, more compensation is needed to incentivize it). 
When \(f\) is additive, $\alpha_i$ simplifies to $c_i/f(\{i\})$.

\paragraph{The Design Problem and Competitive Ratio.}  

We can now describe the full setting: 
Upon the arrival of agent \(i\), the principal learns the cost \(c_i\), and can evaluate the reward function \(f\) on any subset of the arrived agents  $N_i$. Based on this information, the principal must immediately decide whether to hire agent \(i\), without knowledge of future arrivals (but with the ability to revoke earlier choices under free dismissal). 
When agent $i$ is hired, the principal commits to $i$'s share $\alpha_i$ of the eventual reward $f(S)$, where $\alpha_i$ is defined in Eq.~\eqref{eq: def alpha}. This guarantees the agent an incentive to work (despite incurring cost $c_i$) if they make it to the final team $S$. The share $\alpha_i$ is parameterized by $S$ until step $n$, when $S$ is finalized. Once an agent is dismissed, their share is set to zero. 

The players' utilities are as follows: The utility of each agent \(i \in S\) is \(\alpha_i f(S) - c_i\), the utility of any other agent is $0$, and the principal’s utility depends on the remaining share, and is given by 
$g(S):=\bigl(1 - \alpha(S)\bigr) f(S).$

Since the agents' shares $(\alpha_i)_{i\in S}$ are determined by $S$, the principal's design problem boils down to an online algorithm for selecting $S$ such that the principal's utility $g(S)$ is maximized. In online optimization, the performance of an online algorithm is typically evaluated using the \emph{competitive ratio} (CR), which compares the outcome of the online algorithm to that of an optimal offline algorithm that has complete knowledge of the input sequence in advance. 

Formally, an online algorithm $\Alg$ has a CR $\gamma \in [0, 1]$ if for every input sequence $I$, its output satisfies $\Alg(I) \geq \gamma \cdot \OPT(I)$, where $\OPT(I)$ denotes the value achieved by the optimal offline solution. When the algorithm is randomized and the adversary is oblivious (i.e., cannot adapt to the random choices), the \emph{expected competitive ratio} is a random variable defined similarly, requiring that $\mathbf{E}[\Alg(I)] \geq \gamma \cdot \mathrm{\OPT}(I)$ for all $I$. In our model, $\OPT(N_i)=\max_{S_i\subseteq N_i} g(S_i)$ for every $i\in[n]$.

\section{Toolbox: Balance Points}
\label{sec:balance-points}

In this section, we lay the groundwork for the randomized online algorithm presented in Section~\ref{sec:rand-alg-additive}.
In particular, we develop the notion of \emph{balance points}—a key concept that captures the trade-off between allocating additional shares to agents to incentivize work, and preserving the principal’s utility—forming the basis for both the design and analysis of our algorithm in Section~\ref{sec:rand-alg-additive}.
Formal proofs of the results in this section appear in Appendix \ref{app:bp-analysis}. 

\subsection{Agent Quality}

We introduce a quality measure for agents and teams, which captures economic efficiency: 
how much value is generated by the hired agents, relative to the total share that must be allocated to them in order to incentivize their work. 

\begin{definition}[Quality]
\label{def:quality-set}
    Consider an additive reward function $f$. For every team \(S \subseteq N\), its \emph{quality} $q(S)$ is the ratio between the reward it yields, and the total share allocated to the team members:
\begin{align}
\label{eq:quality-set}
    q(S) := \frac{f(S)}{\alpha(S)} = \frac{f(S)}{\sum_{i \in S} \frac{c_i}{f(\{i\})}};~~~q(\emptyset) := 0.
\end{align}
\end{definition}

Definition~\ref{def:quality-set} applied to a singleton team \(\{i\}\) gives the quality of agent $i$: \(q_i := \frac{f(\{i\})}{\alpha_i} = \frac{f^2(\{i\})}{c_i}\). 
Using~Eq.~\eqref{eq:quality-set} we can rewrite the principal's utility using the notion of quality: 
\begin{equation}
\label{eq: g}
    g(S) = (1 - \alpha(S)) \alpha(S) q(S).
\end{equation}

\paragraph{Notation.}

We partition the set of all agents \(N\)
into subsets of uniform quality, ordered from highest to lowest: \(N = Q_1 \cup \dots \cup Q_p\), where for every \(1 \leq x \leq p\), all agents in \(Q_x\) have the same quality, which is strictly greater than that of any agent in \(Q_y\) with \(x < y \leq p\). 
Each \(Q_x\) is referred to as a \emph{quality group}, and we denote its uniform quality by \(q^x\); thus, \(q^x = q_i\) for every agent \(i \in Q_x\) and \(q^x = q(Q_x)\). 

Since $f$ is additive, each agent $i$ has a fixed value $\alpha_i$ that can be determined upon arrival. Consequently, an agent’s quality and associated quality group remain unchanged throughout the process. This stability is a defining property of additive reward functions, and we utilize it in the remainder of this section.

Consider a hired set of agents $S \subseteq N$.
For every \(1 \leq x \leq p\), let $S^x := S \cap \bigcup_{1 \leq y \leq x} Q_y$ be the subset of \(S\) consisting of agents from the first \(x\) quality groups, with \(S^0 := \emptyset\).
 
Let $\bar{S}$ denote the ordering of agents in $S$ sorted in decreasing order of quality. In case of ties in quality, agents with smaller shares appear earlier in the ordering. 
Let $\bar{N}$ denote the entire set of agents, ordered in this way.
For every $i \in [n]$, let 
$\bar{N}_i=(\bar{N})_i$ denote the agents in $\bar{N}$ appearing up to agent $i$, including $i$ itself.
Now we can define $\bar{S}_i:=S \cap \bar{N}_i$. 
In words, these are the agents in $\bar{S}$ that are ordered before agent~$i$.%
\footnote{For example, if the agents arrive in increasing order of quality, and $S=N$, then $\bar{S}_n=\{n\}$ since all other agents succeed agent $n$ in the ordering.}

\subsection{Balance Point Definition}

In this section we define balance points.
Intuitively, a balance point is defined per quality group $Q_x$ and per team $S$ (which determines the set $S^{x-1}$ of higher-quality agents who are allocated shares).
It quantifies the maximum additional share that can be allocated to agents of quality $Q_x$ before the principal’s utility begins to decline. 
Once the balance point is exceeded by the allocated shares, adding more agents from $Q_x$ or from lower-quality groups can only decrease the principal’s utility. Therefore, identifying and maintaining allocations near the balance point is key to achieving optimal utility when selecting agents in an online manner.

\begin{definition}[Balance point]
    \label{def: bp}
    Given a hired set of agents \(S \subseteq N\) and a quality group \(Q_x\), 
    the \emph{balance point} \(b(S, x) \in (0,1)\) of \(Q_x\) relative to \(S\) is a share satisfying the following property. For
    any two subsets \(S', S'' \subseteq Q_x\) such that 
    %\[
    $\big|\alpha(S^{x-1} \cup S') - b(S, x)\big| \leq \big|\alpha(S^{x-1} \cup S'') - b(S, x)\big|$,
    %\]
    it holds that
    \[
    g(S^{x-1} \cup S') \geq g(S^{x-1} \cup S'').
    \] 
\end{definition}

\begin{figure}
    \centering
    \includegraphics[width=1\linewidth]{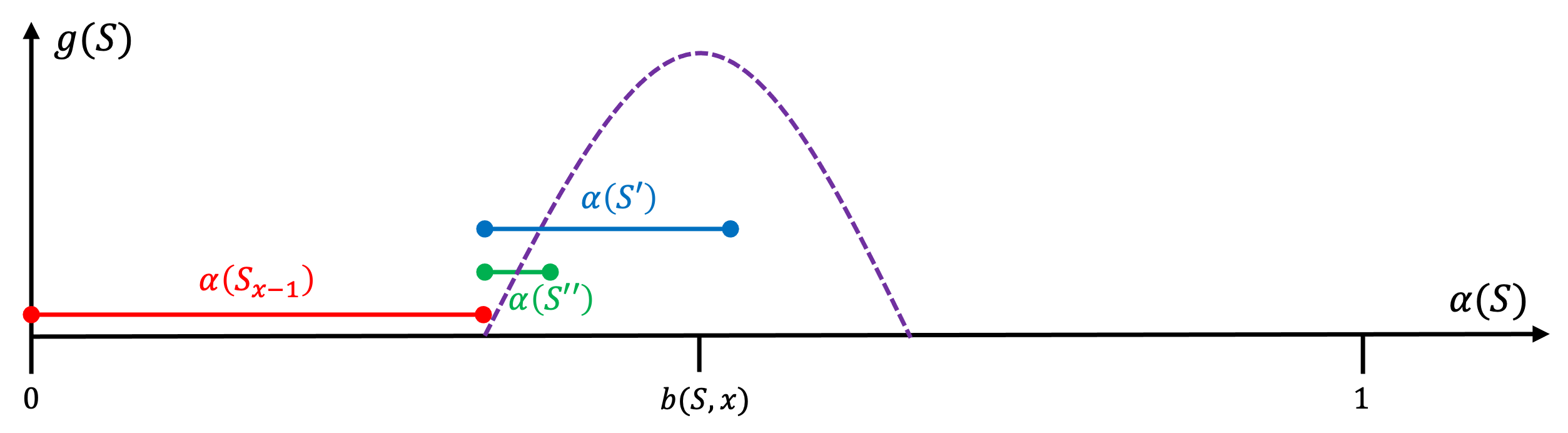}
    \caption{{\bf Geometric intuition for balance points.} 
    For a team $S$ and quality level $x$, the balance point $b(S,x)$ is marked on the $x$-axis, which represents allocated share. Recall that $S^{x-1}$ denotes the subset of agents of quality at least $q^{x-1}$. Consider two subsets $S', S'' \subseteq Q_x$ of quality $q^x$. The horizontal segments show the shares allocated to $S^{x-1} \cup S'$ (red and blue) and $S^{x-1} \cup S''$ (red and green). 
    Since the share of $S^{x-1} \cup S'$ is closer to the balance point than that of $S^{x-1} \cup S''$, by Definition~\ref{def: bp} this implies $g(S^{x-1} \cup S') \geq g(S^{x-1} \cup S'')$. The $y$-axis represents the principal's utility, and the dashed parabola depicts the utility as a function of shares allocated to agents of quality $q^x$. The principal's utility is maximized at the balance point, and indeed $g(S^{x-1} \cup S') \geq g(S^{x-1} \cup S'')$.
}
    \label{fig: Balance Points}
\end{figure}

Figure \ref{fig: Balance Points} illustrates this definition geometrically.
In Lemma~\ref{lem: balance point} we show that such a balance point exists for every quality group and every set of hired agents. 

To better understand the idea of balance points, consider a fixed quality group, \(Q_x\). 
Similarly to Eq.~\eqref{eq: g}, we can use the additivity of $f$ to write the principal's utility as
\begin{align*}
g(S) = \big(1 - \alpha(S)\big) \Big( \alpha(S \cap Q_x) \cdot q^x + \alpha(S \setminus Q_x) \cdot q(S \setminus Q_x) \Big).
\end{align*}
Suppose that the agents hired from other quality groups \(S \setminus Q_x\) are fixed. As the principal allocates more shares to agents in \(Q_x\), the utility may increase at first but is guaranteed to decline eventually.
This trade-off arises because, by the monotonicity of \(f\), hiring additional agents from \(Q_x\) (increasing \(\alpha(S \cap Q_x)\)) increases the reward \(f(S \cap Q_x) = \alpha(S \cap Q_x) q^x\) contributed by this quality group, thereby increasing the total reward.  
At the same time, the principal’s share of the reward \(1 - \alpha(S \setminus Q_x) - \alpha(S \cap Q_x)\) is reduced. 
Thus, the marginal effect of hiring agents in \(Q_x\) on \(g(S)\) follows a concave trajectory as a function of \(\alpha(S \cap Q_x)\), forming a downward-opening parabola with the balance point marking its peak---see Figure~\ref{fig: Balance Points}.

\paragraph{Characterizing Balance Points.}

We now give an alternative characterization of balance points. 
To derive this characterization, it is convenient to introduce an \emph{auxiliary agent}: a hypothetical agent of quality $q^x$, whose share is set according to Eq.~\eqref{eq: def alpha} by setting a corresponding cost.
The auxiliary agent whose share maximizes utility relative to $S^{x-1}$ yields the balance point $b(S,x)$:

\begin{observation}
\label{obs: aux agent bp}
Let $r$ be an auxiliary agent with $q_r = q^x$ such that $g(S^{x-1} \cup \{r\}) \geq g(S^{x-1} \cup \{v\})$ for any other auxiliary agent $v$ with $q_v = q^x$. Then, $\alpha(S^{x-1} \cup \{r\}) = b(S,x)$.
\end{observation}

The concept of auxiliary agent provides a simple and intuitive way to derive the balance point of each quality group, as it represents the optimal share allocation for agents of a given quality. 
This allows us to express \(b(S,x)\) using a concise algebraic formula. 
\begin{lemma}
    \label{lem: balance point}
    For every hired set of agents \(S \subseteq N\) and quality group \(Q_x\),
    \begin{equation}
    \label{eq: bal_point}
    b(S, x) = \frac{1}{2} + \frac{1}{2}\left(1 - \frac{q(S^{x-1} )}{q^x}\right)\alpha(S^{x-1} ).
    \end{equation}
\end{lemma}

\begin{corollary}
\label{rem: bp1-main}
For every \(S \subseteq N\), the balance point of the first quality group is
$b(S,1) = \frac{1}{2}$. 
\end{corollary}

\subsection{Utilizing Balance Points}

We now examine key properties of balance points that will be used in the analysis of our algorithm for additive OMAC (see Section~\ref{sec:rand-alg-additive}). 

\paragraph{Crossing a Balance Point.}
Once a balance point \(b(S, x)\) is reached by agents from \(S^x\), adding agents of equal or lower quality can only reduce the principal's utility. 
Specifically, if an agent in \(Q_x\) is the first in the ordering to cause the cumulative share \(\alpha(S^x)\) to cross \(b(S,x)\), then any additional agent in \(Q_x\) or lower quality groups would decrease \(g(S)\). 
This idea is formalized in the next result.

\begin{lemma}
    \label{lem: one crossing bp}
    Let $i$ be the first agent \(i \in Q_x\) in $S \subseteq N$ such that \(\alpha(\bar{S}_i) \geq b(S, x)\). Then, $g(\bar{S}_i) \geq g(S)$.
\end{lemma}

This result implies that in an optimal hiring strategy, no agent appearing after such a crossing point (in the quality order) would be hired. This property is a cornerstone of our approach, as it ensures that by avoiding balance point crossings, we risk losing at most one potentially valuable agent at any point in time.  
For example, if an optimal algorithm hires an agent from the first quality group whose share size is $\geq \tfrac{1}{2}$, it would not hire any additional agents. This is by applying Corollary~\ref{rem: bp1-main} and Lemma~\ref{lem: one crossing bp}.

\paragraph{Tracking Balance Points Over Time.}

We now show that the role of balance points as natural stopping criteria remains stable throughout the entire online hiring process,
by examining how balance points evolve as agents are hired sequentially. 
Observe that as additional agents are hired, the cumulative shares within each quality group grow, causing the corresponding balance points \(b(S, x)\) to decrease. 
This monotonic behavior is crucial, as it ensures that the reasoning about crossing a balance point in Lemma~\ref{lem: one crossing bp}, established in the offline context, remains valid in the online process. 
In particular, once an agent crosses a balance point, it will continue to do so as more agents are hired. 
The next lemma formalizes this property.

\begin{lemma}
    \label{lem: comp bp 1}
    Let \(S, S' \subseteq N\) be two sets of hired agents, and let \(1 \leq z \leq p\) be an index such that for every quality group \(Q_x\) with \(1 \leq x \leq z\), the shares satisfy \(\alpha(S \cap Q_x) \leq \alpha(S' \cap Q_x)\). 
    Then, for all \(1 \leq x \leq z\), it holds that \(b(S, x) \geq b(S', x)\).
\end{lemma}

Intuitively, Lemma~\ref{lem: comp bp 1} confirms that increasing the total share distributed up to some quality group only reduces the balance points for all lower-quality groups, making future additions less beneficial. 

\paragraph{Comparing Hired Agent Sets Based on Balance Points.}
The notion of balance points allows us to compare different sets of hired agents in terms of the total utility they induce for the principal. 
In particular, if one set allocates shares more effectively across quality groups while staying within the corresponding balance points, its total utility is guaranteed to be higher.

\begin{lemma}
    \label{lem: comp bp 2}
    Let \(S, S' \subseteq N\) be two sets of hired agents. 
    Suppose that for every quality group \(Q_x\), \(1 \leq x \leq p\), the shares satisfies
    $\alpha(S \cap Q_x) \leq \alpha(S' \cap Q_x)$ and $\alpha(S^x) \leq \alpha((S')^x) \leq b(S', x)$,
    then it holds that \(g(S) \leq g(S')\).
\end{lemma}

The statement of the above lemma naturally extends to the case where we introduce an auxiliary agent~$r$ that maximizes the contribution of a given quality group, assuming the agents in all preceding quality groups are fixed. 

\begin{lemma}
    \label{lem: comp bp 3}
    Let \(S, S' \subseteq N\) be two sets of hired agents and \(1 \leq y \leq p\). 
    Suppose that for every quality group \(Q_x\), \(1 \leq x \leq y\), the shares satisfy
    $\alpha(S \cap Q_y) \leq \alpha(S' \cap Q_y)$ and $\alpha(S^y) \leq \alpha((S')^y) \leq b(S', y)$. Then $g(S) \leq g((S')^y \cup \{r\})$,
    where \(r\) is the auxiliary agent that maximizes utility over the fixed \((S')^y\).
\end{lemma}

\section{Positive Result: Randomized Algorithm}
\label{sec:rand-alg-additive}

In this section, we present our randomized algorithm for OMAC.
We start with two deterministic algorithms representing complementary approaches to OMAC with an additive reward function. Our final algorithm is a randomized combination of these strategies.

Algorithm {\sf BP} (Balance Point) hires agents with the highest possible utility while ensuring no group exceeds its balance points. When comparing two agents of equal quality, {\sf BP} prefers the one with the smaller share. To achieve this, it runs {\sf BP-STEP} (Algorithm~\ref{algo: bp}) upon each agent's arrival.

\begin{algorithm}
    \caption{\sf BP-STEP}
    \label{algo: bp}
    \begin{algorithmic}[1]
        \REQUIRE The algorithm receives as input the previously hired group of agents $S$, and an arriving agent~$j$.
        \STATE $S' \gets S \cup \{j\}$
        \STATE Sort the agents in $S'$ in decreasing order of quality, breaking ties by non-decreasing order of shares.
        \STATE Let $\bar{S}' = \{i_1, \dots, i_s\}$  be the sorted list.
        \STATE Set $S \gets \emptyset$
        \FOR{$k = 1$ to $s$}
            \STATE Let $x$ be such that $i_k \in Q_x$
            \IF{ 
            $\alpha(S \cup \{i_k\}) < b(S, x)$}
                \STATE $S \gets S \cup \{i_k\}$
            \ENDIF
        \ENDFOR
        \STATE \RETURN $S$
    \end{algorithmic}
\end{algorithm}

Algorithm {\sf MAX} selects the best agent from all arrivals by running {\sf MAX-STEP} upon each agent's arrival. It is formally described in Algorithm~\ref{alg: max}. For simplicity, we initialize with a dummy agent 0, having $\alpha_0 = 0$, indicating no agents have been hired yet.

\begin{algorithm}
\caption{\sf MAX-STEP} 
\label{alg: max}
    \begin{algorithmic}[1]
        \REQUIRE
        The algorithm receives as input the previously hired agent $i$, and an arriving agent~$j$.
        \RETURN $\arg\max\{g(\{i\}), g(\{j\})\}$
    \end{algorithmic}
\end{algorithm}

We define the randomized algorithm {\sf ALG-OMAC} (Algorithm \ref{alg:omac}) as follows:  
run either {\sf BP} or {\sf MAX} with equal probability (i.e., uniformly at random), throughout the entire game. 

\begin{algorithm}
    \caption{{\sf ALG-OMAC}}
    \label{alg:omac}
    \begin{algorithmic}[1]
        \REQUIRE The algorithm receives an instance of OMAC.
        \STATE With probability $1/2$, run {\sf BP}; with probability $1/2$, run {\sf MAX.}
    \end{algorithmic}
\end{algorithm}

The performance guarantee of {\sf ALG-OMAC} stems from the complementary strengths of the two deterministic algorithms it leverages. Algorithm {\sf BP} is guided by a greedy strategy: it consistently attempts to approach balance points, as described above.  
However, identifying when the algorithm should cross a balance point is challenging. Thus, {\sf BP} adopts a cautious approach—getting as close as possible to the balance points without ever exceeding them.  
This conservatism comes at a cost in scenarios where crossing a balance point would lead to greater utility.  
Nevertheless, this potential loss is mitigated by the performance of {\sf MAX}, which simply selects the best single agent.  
That is, the utility gained from hiring the best agent alone often compensates for any utility missed by {\sf BP} due to its conservative behavior. Building on this intuition, we now prove the main theorem establishing the competitive ratio of {\sf ALG-OMAC}:
\begin{theorem}
        \label{theo: A CR}
    Given an instance $I$ of OMAC, it holds that $\OPT(I) \leq g({\sf BP}(I)) + g({\sf MAX}(I))$.
\end{theorem}
\begin{corollary}
    The expected CR of {\sf ALG-OMAC} is at least $\frac{1}{2}$.
\end{corollary}
\begin{proof}
By the definition of {\sf ALG-OMAC}, its expected utility is given by  
\[
\frac{1}{2} \left( g({\sf BP}(I)) + g({\sf MAX}(I)) \right).
\]  
Therefore, by Theorem~\ref{theo: A CR}, the competitive ratio of the algorithm is at least $\frac{1}{2}$.
\end{proof} 

We proceed to formally analyze the algorithm.
Let $S^1$ and $S^2$ denote the groups of agents hired by {\sf BP} and {\sf MAX}, respectively, given the arrival sequence $N$.  
Let $\bar{S}^* = \{1^*, \dots, k^*\}$ denote an optimal solution of size $k = |S^*|$, where agents are sorted in decreasing order of quality, with agents of the same quality ordered in non-decreasing order of shares.

We now identify an agent that will play a central role in the analysis of our algorithm. Let agent $t \in Q_z$ be the first agent that satisfies at least one of these two properties:
\begin{itemize}
    \item $\alpha(\bar{N}_t) \geq b(N, z)$ and $\alpha_t \leq 1/2$;
    \item $t = k^*$.
\end{itemize}
In other words, $t$ is defined as the first agent in the ordering $N_t$ who, when hired together with all preceding agents, crosses their corresponding balance point and whose share size is at most $\frac{1}{2}$—or, alternatively, as the last agent in $\bar{S}^*$, whichever occurs earlier. We denote by $t'$ the agent that immediately precedes $t$ in the ordering.

\begin{lemma}
    \label{lem: m in s1}
     $\bar{N}_{t'} \subseteq S^1$.
\end{lemma}
\begin{proof}
Consider an arriving agent $i \in \bar{N}_{t'}$ and a group of hired agents $S \subseteq N$ such that $i \notin S$.  
We aim to show that algorithm {\sf BP} does not dismiss agent $i$, regardless of the choice of $S$.  
This immediately implies that every agent $i \in \bar{N}_{t'}$ remains hired under {\sf BP}, from the time of their arrival until the end of the sequence $N$, meaning that $i \in S^1$.

Assume towards contradiction that there exists such an agent $i \in \bar{N}_{t'}$ and a group of hired agents $S \subseteq N$ such that $i \notin S$, and {\sf BP} dismisses $i$ upon arrival. We distinguish between two cases. 

If $\alpha_i > \frac{1}{2}$, then by Corollary~\ref{rem: bp1-main},
this agent crosses its matching balance point in $S^*$. Consequently, by Lemma~\ref{lem: one crossing bp} and the optimality of the matching, it must be that $i^* = k^*$. We get that, by the definition of $t$, $i$ succeeds $t$ in the ordering, contradicting the fact that $i \in \bar{N}_{t'}$.

Now, suppose that $\alpha_i \leq \frac{1}{2}$. Let $i \in Q_y$.  
Since $i$ is not hired when considered against the group $S$, by the definition of {\sf BP},
\[
\alpha(\bar{S}'_i) \geq b(S', y),
\]
where $S'$ denotes the output of {\sf BP} in this iteration.

By applying Lemma~\ref{lem: comp bp 1}, we have
\[
b(S', y) \geq b(N, y).
\]
Combining this with the fact that $\bar{S}'_{i-1} \subseteq \bar{N}_{i-1}$, we get:
\[
\alpha(\bar{N}_{i-1} \cup \{i\}) \geq \alpha(\bar{S}'_{i-1} \cup \{i\}) \geq b(S', y) \geq b(N, y).
\]
By the definition of $t$, we have that $i$ succeeds $t$ in the ordering, leading to a contradiction.

It follows that agent $i$ remains hired until the sequence ends, i.e., $i \in S^1$ for all $i \in \bar{N}_{t'}$.
\end{proof}

We note that containment (between two subsets $S,S' \subseteq N$) is a special case of Lemma \ref{lem: comp bp 2}. Hence, using Lemma~\ref{lem: m in s1} and the fact that $\bar{S}^*_{t'} \subseteq \bar{N}_{t'}$, we have 
\begin{equation}
\label{eq:upper_lower_bound_Mt}
g(\bar{S}^*_{t'}) \leq g(\bar{N}_{t'}) \leq g(\bar{S}^1).
\end{equation}

We can now complete the analysis of {\sf ALG-OMAC}.

\begin{proof}[\bf Proof of Theorem \ref{theo: A CR}:]
We distinguish between two cases.

Assume first that $t$ is the last agent in $S^*$. By definition, $S^* = \bar{S}^*_{t'} \cup \{t\}$, and by the sub-additivity of $g$, we have
\[
g(S^*) = g(\bar{S}^*_{t'} \cup \{t\}) \leq g(\bar{S}^*_{t'}) + g(\{t\}).
\]
By inequality~\eqref{eq:upper_lower_bound_Mt}, $g(\bar{S}^*_{t'}) \leq g(S^1)$, and from the maximality of the single agent selected by algorithm {\sf MAX}, we have that $g(\{t\}) \leq g(S^2)$. Therefore,
\[
g(S^*) \leq g(\bar{S}^*_{t'}) + g(\{t\}) \leq g(S^1) + g(S^2),
\]
as required.

\medskip

If $t$ is not last in $S^*$ then it must hold that
 $\alpha_t \leq \frac{1}{2}$. We know that ${S^*}^{z-1} \subseteq N^{z-1}$, which allows us to apply Lemma~\ref{lem: comp bp 3}. Consequently,
$g(S^*) \leq g(N^{z-1} \cup \{r\})$,
where $r$ denotes the auxiliary agent corresponding to $N^z$. By sub-additivity of $g$, we have that
$g(S^*) \leq g(N^{z-1}) + g(\{r\})$.
Furthermore, using Lemma~\ref{lem: comp bp 2}, we have $N^{z-1} \subseteq \bar{N}_{t'}$, implying that
$g(N^{z-1}) \leq g(\bar{N}_{t'})$. Thus, $g(S^*) \leq g(\bar{N}_{t'}) + g(\{r\})$.

Now, observe that in this case, agent $t$ crosses its balance point over $\bar{N}_{t'}$, i.e.,
\[
\alpha(\bar{N}_t) \geq b(N, z).
\]
As a result, we have by Corollary~\ref{rem: bp1-main} 
\[
\alpha(\bar{N}_{t'}) + \alpha_t = \alpha(\bar{N}_{t'} \cup \{t\}) \geq b(N, z) = \alpha(\bar{N}_{t'}) + \alpha_r,
\]
and hence $\alpha_t \geq \alpha_r$.

Since $\alpha_r \leq \alpha_t \leq \frac{1}{2}$, it follows that $g(\{r\}) \leq g(\{t\})$. Moreover, because $t \in N$ is not an auxiliary agent (unlike $r$), {\sf MAX} ensures that $g(\{t\}) \leq g(S^2)$. Therefore,
$g(\{r\}) \leq g(S^2)$.

Finally, from inequality~\eqref{eq:upper_lower_bound_Mt}, we also have $g(\bar{N}_{t'}) \leq g(S^1)$. Combining these bounds, we conclude that
\[
g(S^*) \leq g(\bar{N}_{t'}) + g(\{r\}) \leq g(S^1) + g(S^2).
\]
This completes the proof.
\end{proof}

\section{Negative Results}
\label{sec:negative}
This section delineates the boundaries of tractability for online contract design by establishing several impossibility results. We begin with additive rewards in Subsection~\ref{sec:negative-additive}, complementing our positive results by showing that randomization is necessary and that a $1/2$ competitive ratio is optimal. We then turn to XOS rewards in Subsection~\ref{sec:submodular}, where we identify the threshold at which rich combinatorial interactions preclude any constant-factor approximation. Finally, Subsection~\ref{app:dismiss-free} investigates the structural role of preemption, demonstrating that the ability to dismiss agents is essential for achieving meaningful performance guarantees in an online adversarial setting.

\subsection{Additive Rewards}
\label{sec:negative-additive}

We identify two fundamental limitations in the additive case, both arising from the same adversarial scenario: the arrival of an agent demanding a large share. Such an agent can either cause the principal to lose all project value, or block the recruitment of more effective agents. Consequently, deterministic strategies might perform arbitrarily poorly.

\begin{theorem}
    \label{theo:add-neg-det}
    The CR of any deterministic algorithm for OMAC can be arbitrarily small.
\end{theorem}

Furthermore, even randomized algorithms cannot fully avoid the trade-off induced by such an agent: no randomized algorithm can achieve a competitive ratio better than $1/2$ (our randomized algorithm in Section~\ref{sec:rand-alg-additive} matches this bound).  

\begin{theorem}
\label{thm:CR_rand_upper_bound}
    No (randomized) algorithm for OMAC can achieve a CR $ > 1/2$.
\end{theorem}

\begin{proof}[Proof sketch]
    Let $\varepsilon \in (0, 1)$ be a small value, $n = \varepsilon^2/2 + 1$ and $N = \{1, \dots, n\}$ such that  $\alpha_1 = 1 - \varepsilon^2, q_1 = \varepsilon$ and $\alpha_i = \varepsilon^2, q_i = \varepsilon^2$ for every $2 \leq i \leq n$. 

 Given a (possibly randomized) algorithm, let $X$ be an indicator random variable for hiring agent 1 upon arrival. For any small agent $i \in N \setminus {1}$, let $T_i \subseteq N_i \setminus {1}$ be the random set of small agents hired by the algorithm among the first $i$ arrivals, conditioned on hiring agent 1, with $t_i = |T_i|$.   

    Then every randomized algorithm belongs to one of the following types:
    \begin{enumerate}
    \item ${\sf ALG}_1$ is the class of algorithms for which the corresponding indicator variable $X$ satisfies $\Pr(X) \leq 1/2$.

    \item For all $i \in N \setminus \{1\}$, let ${\sf ALG}_i$ be the class of algorithms for which $\Pr(X) > 1/2$. In addition, $i$ is the first index such that $\mathbf{E}[t_i] > 1$, meaning that for every $1 \leq j < i$, we have $\mathbf{E}[t_j] \leq 1$.
    
    \item ${\sf ALG}_{n+1}$ denotes the class of algorithms for which $\Pr(X) > 1/2$, and $\mathbf{E}[t_i] \leq 1$ for all $1 \leq i \leq n$.
    \end{enumerate}
Consider one of the classes ${\sf ALG}_i$ for $1 \leq i \leq n+1$. We claim that by setting $k = \min\{i, n\}$, the utility of every algorithm in this class on the input prefix $N_k$ can be made arbitrarily small compared to the optimal solution. In Appendix~\ref{app:neg-add}, we prove this claim separately for each class.
\end{proof}

\subsection{XOS Rewards}
\label{sec:submodular}

We study OMAC in a broader setting where the reward function $f$ may be XOS. 
In this generalized model, for every $S \subseteq N$ and $i \in S$, the contribution of agent $i$ may depend on the entire set $S$. As a result, we can no longer fix an agent’s share $\alpha_i$ upon arrival, since this value may vary as other agents are hired or dismissed. This dependency makes the problem substantially harder in adversarial settings.

We show that even in the free-dismissal model, where the principal may withdraw previously offered contracts, no (possibly randomized) algorithm can guarantee a constant competitive ratio. This difficulty stems from the adversary’s power to determine that any hired agent may lose its marginal contribution when additional agents arrive, whereas dismissed agents may retain their potential value. Analogous to real-world economic settings, the principal lacks foresight into how an arriving agent will interact with future ones, rendering reliable long-term decisions infeasible.

\begin{theorem}
    \label{theo: free-dismissal SM negative}
    For OMAC with an XOS contribution function $f$, the competitive ratio $\gamma$ of any algorithm can be made arbitrarily small.
\end{theorem}

\paragraph{Proof Idea.}
We apply Yao’s min–max principle and construct a distribution $\mathcal{D}$ over input instances that defeats every deterministic algorithm $\mathsf{ALG}$. Consider a sequence of $n$ seemingly identical agents, each with the same positive action cost and the same individual contribution when hired alone. However, these agents are ineffective in teams: any group of two or more hired agents yields no more utility than a single one. Formally, for any subset $S$ of size at least two, $f(S) = f(\{i\})$ for some $i \in S$. Thus, the marginal contribution of any agent in such a group is zero, implying $\alpha_i = \infty$. The utility of hiring one agent is constant, but hiring a group results in utility $-\infty$ due to the equilibrium constraint.

Since the agents are indistinguishable, $\mathsf{ALG}$ has no better option than hiring one arbitrarily. However, the distribution $\mathcal{D}$ selects exactly one of these agents to be ``good”, chosen uniformly at random, while the rest are “bad.” Bad agents remain ineffective regardless of who else is hired. In contrast, the good agent has a latent property: once hired, it enables productive collaboration with agents from future groups.

After this first group, we present $\mathsf{ALG}$ with a second group of $n$ agents, again indistinguishable in cost and standalone contribution. Only one among them is good (chosen uniformly at random), and it can collaborate effectively with the good agent from the first group—but not with any bad agent. Lacking any way to identify the good agent, $\mathsf{ALG}$ can hire at most one. If it previously hired the good agent, the new hire increases utility; otherwise, no gain is possible.

At this point, $\mathsf{ALG}$ faces a dilemma: it must either retain its current hire (possibly a bad agent) or dismiss it in favor of a new, equally uncertain one. In either case, if it failed to identify the good agent in the first group, the best it can achieve is half the utility that would be obtained by correctly pairing the two good agents.

We repeat this process for $m$ groups. In each round, the good agent is selected uniformly at random, and $\mathsf{ALG}$ receives $n$ indistinguishable candidates. The probability of selecting the good agent in a group is $1/n$, so with high probability, the algorithm continues to miss key opportunities. As $m$ grows, the adversary’s ability to suppress utility increases, and the expected competitive ratio of $\mathsf{ALG}$ tends toward zero.

To further clarify the construction of the distribution $\mathcal{D}$ and the behavior of algorithms in this setting, we give Example~\ref{ex: neg sm}, which illustrates the core idea with parameters $n = 2$ and $m = 3$.

\begin{example}
\label{ex: neg sm}
Let \(N = \{B_1, G_1, B_2, G_2, B_3, G_3\}\), where \(B_j\) and \(G_j\) denote, respectively, the bad and good agents of group \(j\). Each agent has cost \(c_i = \varepsilon \in \mathbb{R}_{>0}\).  
For \(j = 1, 2\), the good agent \(G_j\) is selected uniformly at random from within its group, with the other designated as bad.  

Define the reward function \(f\) as:
\[
f(S)=
\begin{cases}
1,
& B_1\in S, \\[4pt]

\mathbf{1}_{G_1\in S}+1,
& B_1\notin S,\; B_2\in S, \\[4pt]

\mathbf{1}_{G_1\in S}+\mathbf{1}_{G_2\in S},
& B_1,B_2\notin S,\\
& S\cap\{B_3,G_3\}=\emptyset, \\[4pt]

\mathbf{1}_{G_1\in S}+\mathbf{1}_{G_2\in S}+1,
& B_1,B_2\notin S,\\
& S\cap\{B_3,G_3\}\neq\emptyset.
\end{cases}
\]

Consider a deterministic algorithm \({\sf ALG}\).  
To avoid utility of \(-\infty\), it must never hire a bad agent together with any other agent. Thus, when presented with the first pair \(\{B_1, G_1\}\), \({\sf ALG}\) must hire exactly one agent. Similarly, upon arrival of the second pair \(\{B_2, G_2\}\), it again hires exactly one, discarding the first agent if necessary (e.g., replacing \(B_1\), whose presence contributes no additional utility).

Upon observing the final pair \(\{B_3, G_3\}\), the algorithm again selects exactly one agent and may dismiss a previously hired one to avoid infeasible configurations.

Let us compute the expected utility. The probability that \({\sf ALG}\) hires both \(G_1\) and \(G_2\) is \(\frac{1}{2} \cdot \frac{1}{2} = \frac{1}{4}\), resulting in utility close to \(3\). The probability that both hires from the first two pairs are bad is also \(\frac{1}{4}\), yielding utility near \(1\). In the remaining \(\frac{1}{2}\) of cases, exactly one of \(G_1\), \(G_2\) is hired, leading to utility close to \(2\). Thus, the expected utility of \({\sf ALG}\) over \(I \sim \mathcal{D}\) is approximately:

\[
\mathbf{E}_{I \sim \mathcal{D}}[g({\sf ALG}(I))] 
= \tfrac{1}{4}(1 - 3\varepsilon) 3
+ \tfrac{1}{2}(1 - 2\varepsilon) 2
+ \tfrac{1}{4}(1 - \varepsilon)
\approx 2.
\]

In contrast, the optimal policy—hiring both good agents \(G_1, G_2\) and any one of \(\{B_3, G_3\}\)—achieves a utility of approximately \(3\). Hence, by Yao’s min–max principle, the competitive ratio of any randomized algorithm over $\mathcal{D}$ is at most \(\frac{2}{3}\).
\end{example}

Before presenting the formal proof of Theorem~\ref{theo: free-dismissal SM negative}, we state the following auxiliary lemma (proved in Appendix~\ref{app: rem SM neg 1}), which captures a key recurrence used in the analysis:

\begin{lemma}
\label{lem: aux rem SM neg}
For every \(1 \leq h, \ell \leq m\) and \(n \in \mathbb{N} \setminus \{0\}\), define:
\[
A_h^\ell = 
\begin{cases}
  \displaystyle \frac{h}{m}, & \text{if } \ell = m, \\[6pt]
  \displaystyle \frac{1}{n} A_{h+1}^{\ell+1} + \frac{n-1}{n} A_h^{\ell+1}, & \text{if } 1 \leq \ell < m.
\end{cases}
\]
Then it holds that \(A_1^1 \leq \frac{1}{n} + \frac{1}{m}\).
\end{lemma}

\begin{proof}[\bf Proof of Theorem \ref{theo: free-dismissal SM negative}:]
To show that the competitive ratio $\gamma$ can be made arbitrarily small, we apply Yao’s min-max principle:
\begin{align*}
    \gamma = \max_{\text{randomized } {\sf ALG}_r} \min_I \frac{\mathbf{E}[{\sf ALG}_r(I)]}{\OPT(I)} \leq 
    \min_{\mathcal{D}} \max_{\text{deterministic } {\sf ALG}} 
    \mathbf{E}_{I \leftarrow \mathcal{D}} \left[ \frac{{\sf ALG}(I)}{\OPT(I)} \right].
\end{align*}

It thus suffices to construct a distribution $\mathcal{D}$ such that
\[
\mathbf{E}_{I \leftarrow \mathcal{D}} \left[ \frac{{\sf ALG}(I)}{\OPT(I)} \right] \to 0
\quad \text{for every deterministic } {\sf ALG}.
\]

Let $n, m \in \mathbb{N}$ and $\varepsilon \in \mathbb{R}_{>0}$ be a small constant (assume $\varepsilon \ll \frac{1}{m}, \frac{1}{n}$).  
We define the agent set as \(N = M_1 \cup \dots \cup M_m\), where each group \(M_k\) contains $n$ agents labeled \(i_k\) for \(i \in [n]\), and the groups arrive sequentially in order of increasing $k$.

An input vector \(\sigma \in M_1 \times \dots \times M_{m-1}\) is a tuple \((\sigma_1, \dots, \sigma_{m-1})\) such that each \(\sigma_k \in M_k\) is a designated ``good" agent.

For any subset \(S \subseteq N\), let \(k_S\) denote the largest index such that \(M_{k_S} \cap S \neq \emptyset\) and \(M_k \cap S = \emptyset\) for all \(k > k_S\); define \(k_S = 0\) when \(S = \emptyset\).

We now define an input instance \(I_\sigma\) corresponding to each \(\sigma\). For every $k \in [m]$ and $i_k \in M_k$, we define an additive clause $w_{i_k}^\sigma:2^N\to \mathbb{R}_{\ge 0}$, as follows:
$$w_{i_k}^\sigma(S) = \left| S \cap \{\sigma_\ell\}_{\ell=1}^{k-1} \cap \{i_k\} \right|,$$
with a reward function \(f_\sigma : 2^N \to \mathbb{R}_{\ge 0}\) given by:
\[
f_\sigma(S) = \max_{k \in [m], i_k \in M_k}w^\sigma_{i_k}(S).
\]

We define the input distribution \(\mathcal{D}\) to be uniform over all instances \(I_\sigma\) where \(\sigma \in M_1 \times \dots \times M_{m-1}\).

\paragraph{XOS and Monotonicity.}
We note that for every $\sigma$ and  $i_k \in N$, $w_{i_k}^\sigma$ is additive and monotone, assigning some agents a potential contribution of $1$ and $0$ to the others. Thus, by definition (see Section \ref{sec:preliminaries}) $f_\sigma$ is necessarily XOS, and as a maximum of monotone functions, $f_\sigma$ is monotone as well.

\paragraph{Utility Analysis.}
Assume each agent \(i_k\) has cost \(c_{i_k} = \varepsilon\). If every agent in a hired set \(S \subseteq N\) contributes positively, then each has \(\alpha_{i_k} = c_{i_k}\), and using Eq.~\eqref{eq: g}, we compute:
\begin{align}
\label{eq: g-xos-neg}
g(S) = (1 - \alpha(S)) f_\sigma(S) 
= \left(1 - \sum_{i_k \in S} \varepsilon \right) |S| 
\approx |S|.
\end{align}
In contrast, if any agent \(i_k \in S\) does not contribute—i.e., is not a designated agent or included from an invalid group—then \(\alpha_{i_k} = \frac{c_{i_k}}{0} = \infty\), and the total utility becomes \(g(S) = -\infty\).

Thus, the algorithm must avoid non-contributing agents entirely, which the adversary exploits: since the designated agent in each group is chosen uniformly at random (from the perspective of \({\sf ALG}\)), the probability of consistently choosing contributing agents decays over rounds. This compounding uncertainty leads to arbitrarily small expected utility relative to the optimum, implying that the competitive ratio \(\gamma\) is unbounded below.

We now formalize the proof idea presented above. Let $A_h^\ell$ denote the expected competitive ratio achieved by the \emph{best} deterministic algorithm, conditioned on having exactly $h$ contributing agents hired after the arrival of the last agent in group $M_\ell$ (and no other agents).  
In particular, our goal is to bound $A_1^1$, which, by Yao's min-max principle, satisfies:
\[
A_1^1 = \max_{\text{deterministic } {\sf ALG}} \mathbf{E}_{\sigma \leftarrow \mathcal{D}} \left[ \frac{{\sf ALG}(I_\sigma)}{\OPT(I_\sigma)} \right].
\]
This equality holds because, at the first step, any optimal deterministic algorithm must hire exactly one agent from $M_1$, and all such agents are indistinguishable ex ante. 

We now define a recursive expression for $A_h^\ell$:
\[
A_h^\ell = 
\begin{cases}
  \displaystyle \frac{h}{m} & \text{if } \ell = m, \\[10pt]
  \displaystyle \frac{1}{n} A_{h+1}^{\ell+1} + \frac{n-1}{n} A_h^{\ell+1} & \text{if } 1 \leq \ell < m.
\end{cases}
\]
We argue that this recursion correctly captures the behavior of $A_h^\ell$.

\paragraph{Base case ($\ell = m$).}
At the arrival of the final group $M_m$, suppose the algorithm has hired $h$ contributing agents so far. Since the utility from each contributing agent is approximately $1$ (the precise value is given in Eq. \eqref{eq: g-xos-neg}), the total expected utility of the algorithm is approximately $h$.  
On the other hand, the optimal solution consists of the $m-1$ designated agents $\{\sigma_1, \dots, \sigma_{m-1}\}$ plus any single agent from $i_m \in M_m$, with $w_{i_m}^\sigma$ as the additive clause that maximizes $f_\sigma$, yielding:
\[
\OPT(I_\sigma) \approx m.
\]
Therefore, the competitive ratio is approximately $\frac{h}{m}$, which matches the definition of $A_h^m$.

\paragraph{Induction step ($1 \leq \ell < m$).}
Consider the best possible decision faced by a deterministic algorithm ${\sf ALG}$ after the arrival of the last agent in group $M_{\ell}$, assuming it has hired a group $S$ of $h$ contributing agents, and consider $w_{i_k}^\sigma$, the clause that realizes $f_\sigma$. The fact that hiring more than one agent from $M_k$ results in a utility of $-\infty$ as described above implies that there must be at least $h - 1$ agents in $S$ that match the designated agents $\{\sigma_1, \dots, \sigma_{k-1}\}$. Consequently, these agents will continue contributing regardless of future hiring decisions.

Let $i_k$ denote the last agent hired by ${\sf ALG}$, where $k \leq \ell$. We now argue that optimality requires $k = \ell$. Suppose instead that $i_k = \sigma_k$ for some $k < \ell$. Then, hiring any agent $i_\ell$ from $M_\ell$ would lead to $w_{i_\ell}^\sigma(S) = |S| + 1$, strictly increasing the utility—contradicting the optimality of ${\sf ALG}$. Therefore, it must be that $i_k \neq \sigma_k$, and it is strictly better for ${\sf ALG}$ to dismiss $i_k$ and hire some agent from $M_\ell$. Doing so preserves the contributions from earlier designated agents while introducing a chance that the new hire is $\sigma_\ell$, potentially increasing utility in the next stage.

Because $\sigma_\ell$ is drawn uniformly at random from $M_\ell$, and its members are indistinguishable upon arrival, the probability that ${\sf ALG}$ selects $\sigma_\ell$ is $\frac{1}{n}$. In that case, after the arrival of the last agent in $M_{\ell+1}$, hiring any agent from that group increases the number of contributing agents to $h+1$, transitioning us to state $A_{h+1}^{\ell+1}$. In the complementary case—where $i_\ell \neq \sigma_\ell$—the algorithm remains with only $h$ contributing agents, even after interacting with $M_{\ell+1}$, yielding state $A_h^{\ell+1}$.

Applying the law of total expectation over these two outcomes, we recover the recurrence:
\[
A_h^\ell = \frac{1}{n} A_{h+1}^{\ell+1} + \frac{n-1}{n} A_h^{\ell+1}.
\]

To conclude, by Lemma~\ref{lem: aux rem SM neg}, we have
\[
A_1^1 = O\left(\frac{1}{n} + \frac{1}{m}\right).
\]
Hence, the competitive ratio $\gamma$ can be made arbitrarily small by choosing sufficiently large $n$ and $m$. This completes the proof.
\end{proof}

\subsection{Without Preemption}
\label{app:dismiss-free}

To complete the picture, we show that dismiss-free OMAC admits no algorithm with a competitive ratio (CR) bounded away from zero. This impossibility holds even under the assumption of uniform agent quality (see Definition~\ref{def:quality-set}).
The impossibility follows by taking $n \rightarrow \infty$ in the next theorem: 

\begin{theorem}
    \label{theo: general negative}
    Consider OMAC with uniform agent quality $q$ and at most $n$ arriving agents. The CR $\gamma$ of any algorithm $\mathsf{ALG}$ satisfies $\gamma = O(1/n)$.
\end{theorem}

\begin{proof}
    To show that $\gamma$ can be made arbitrarily close to zero, we invoke Yao's min-max principle:
    \begin{align*}
    \gamma = \max_{\text{randomized } \mathsf{ALG}_r} \min_I \frac{\mathbf{E}[\mathsf{ALG}_r(I)]}{\OPT(I)} \leq
    \min_{\mathcal{D}} \max_{\text{deterministic } \mathsf{ALG}} \mathbf{E}_{I \leftarrow \mathcal{D}}\left[\frac{\mathsf{ALG}(I)}{\OPT(I)}\right].
    \end{align*}
    Our goal is therefore to construct a distribution $\mathcal{D}$ over input instances such that for every deterministic algorithm $\mathsf{ALG}$, the expected competitive ratio tends to zero.

    Construct the set of agents $N = \{1, \dots, n\}$, where for each $1 \leq i \leq n$, agent $i$ has share size $\alpha_i = 1 - \varepsilon^{n - i}$. Accordingly, the principal’s utility from hiring agent $i$ is given by
    \[
    g(\{i\}) = \varepsilon^{n - i}(1 - \varepsilon^{n - i})q.
    \]

    For each $1 \leq i \leq n$, define $I_i$ to be an input instance consisting of the first $i$ agents in the arrival order. Let $\mathcal{D}$ be the uniform distribution over the set $\{I_1, \dots, I_n\}$.

    Since all agents satisfy $\alpha_i > \frac{1}{2}$, any algorithm that hires more than one agent will incur total share size exceeding $1$, violating feasibility. Thus, the optimal solution in any instance hires only a single agent. Consequently, we focus on deterministic algorithms $\mathsf{ALG}_j$ that hire exactly one agent—specifically, the $j$-th one.

    Over instance $I_j$, the algorithm $\mathsf{ALG}_j$ hires agent $j$ and achieves a competitive ratio of $1$, as this is also the optimal solution. However, for every $i < j$, the algorithm $\mathsf{ALG}_j$ fails to hire any agent (since agent $j$ has not yet arrived), resulting in a competitive ratio of $0$. For $i > j$, although $\mathsf{ALG}_j$ hires agent $j$, this is suboptimal since agent $i$ provides strictly higher utility. Specifically, the ratio of their utilities is
    \[
    \frac{g(\{j\})}{g(\{i\})} = \frac{\varepsilon^{n - j}(1 - \varepsilon^{n - j})q}{\varepsilon^{n - i}(1 - \varepsilon^{n - i})q} = \varepsilon^{i - j} \cdot \frac{1 - \varepsilon^{n - j}}{1 - \varepsilon^{n - i}} = O(\varepsilon^{i - j}),
    \]
    implying that the competitive ratio over $I_i$ is at most $O(\varepsilon^{i - j}) = O(\varepsilon)$.

    Since $\mathsf{ALG}_j$ achieves a non-negligible competitive ratio only on instance $I_j$, and $\Pr[I_j] = 1/n$ under $\mathcal{D}$, we conclude that the expected competitive ratio is
    \[
    \mathbf{E}_{I \leftarrow \mathcal{D}}\left[\frac{\mathsf{ALG}_j(I)}{\OPT(I)}\right] = O\left(\frac{1}{n} + \varepsilon\right).
    \]
    Because $\varepsilon$ is arbitrarily small, the result follows.
\end{proof}

\section{Conclusion}
\label{sec:conclusion}

We introduce the Online Multi-Agent Contracting (OMAC) model, bridging economic contract theory and adversarial online algorithm design. We provide an optimal $1/2$-competitive randomized algorithm for additive rewards, and establish that the problem becomes intractable under reward structures with complementarities.
Our findings demonstrate that incentive constraints introduce complexities absent from classical online optimization. While the problem is related to known preemption settings, standard techniques fall short, motivating the introduction of balance points to align sequential hiring with the principal’s utility. 

Dynamic contract design, like other algorithmic decisions fundamentally shaped by incentive compatibility, requires new tools; we outline several promising directions for future work:

\paragraph{Stochastic Online Models.}
Adversarial analysis provides a robust baseline, but exploring stochastic refinements is a natural next step. Following the standard progression in mechanism design—from adversarial to secretary~\cite{kleinberg2005multiple,babaioff2007matroids} and prophet models~\cite{samuelcahn1984comparison,kleinberg2012matroid}—could yield stronger guarantees and clarify the trade-off between robustness and efficiency in dynamic contracting.

\paragraph{Beyond Additive and XOS Reward.}
Bridging the gap between additive and XOS rewards by investigating intermediate classes—such as submodular functions \cite{dutting2022combinatorial}—remains a compelling frontier. Our reduction to online budgeted maximization in Appendix \ref{app:ks to omac} suggests that as the theory of the online knapsack problem continues to mature, breakthroughs in this regime are increasingly within reach \cite{doron2025algorithm, doronarad2024lower}.

\paragraph{Enhancing the Model.}
Another promising direction is to extend the OMAC framework itself beyond binary actions and linear contracts. Allowing agents to choose among multiple effort levels~\cite{dutting2025multi}, or considering contracts that are non-linear in the project's outcome, would capture a broader range of incentive environments. Additionally, one could model actions as part of the online process rather than fixed solely at evaluation, introducing new layers of strategic interaction.
On the structural side, one can move beyond unrestricted preemption towards more general notions of recourse, as studied in the consistency literature~\cite{Lattanzi2017,Fichtenberger2021,Gupta2020,Gupta2022,Lacki2021,Bhattachary2024}. 
Such extensions would continue to bring the model closer to practical applications and also raise new structural questions about how incentives interact with online decision-making.

\clearpage
\bibliographystyle{plain}
\bibliography{bibliography}
\newpage

\appendix

\begingroup

\setcounter{section}{0} 
\onecolumn

\section{Appendix for Section~\ref{sec:balance-points}}
\label{app:bp-analysis}

\begin{proof}[\bf Proof of Lemma~\ref{lem: balance point}:]
    
Using Observation \ref{obs: aux agent bp}, we focus our attention on finding the auxiliary agent \(r\). First, we fix \(r\)'s quality \(q_r = q^x\).
Adding this agent to \(S^{x-1}\) would induce a principal utility of \(g(S^{x-1} \cup \{r\})\). Let us simplify this expression: 
    \begin{align*}
g(S^{x-1} \cup \{r\}) 
&= \left( 1 - \alpha(S^{x-1}) - \alpha_r \right) \left( \alpha_r q^x + \alpha(S^{x-1}) q(S^{x-1}) \right) \\
&= \big(1 - \alpha(S^{x-1}) - \alpha_r\big) \cdot \alpha_r q^x 
    - \alpha_r \alpha(S^{x-1}) q(S^{x-1}) \\
&\quad + \big(1 - \alpha(S^{x-1})\big) \cdot \alpha(S^{x-1}) q(S^{x-1}).
\end{align*}
    The first equality follows from the definition of quality of a group of agents \eqref{eq:quality-set}. Moreover, by the definition of principal's utility, it holds that \(g(S^{x-1}) = (1 - \alpha(S^{x-1}))\alpha(S^{x-1})q(S^{x-1})\). Thus,
    \begin{align*}
g(S^{x-1} \cup \{r\}) 
&= (1 - \alpha(S^{x-1}) - \alpha_r) \cdot \alpha_r q^x 
   - \alpha_r \alpha(S^{x-1}) q(S^{x-1}) + g(S^{x-1}) \\
&= \alpha_r q^x \left( 1 - \alpha(S^{x-1}) \left( 1 + \frac{q(S^{x-1})}{q^x} \right) - \alpha_r \right) + g(S^{x-1}).
\end{align*}

    As \(g(S^{x-1})\) does not depend on \(r\), we get that this utility is maximized when
    \[
    \alpha_r = \frac{1 - \alpha(S^{x-1})\left(1 + \frac{q(S^{x-1})}{q^x}\right)}{2}.
    \]
    Plugging in this value, we have that
    \[
    b(S, x) = \alpha(S^{x-1}) + \alpha_r = \frac{1}{2}(1 + (1 - \frac{q(S^{x-1} )}{q^x})\alpha(S^{x-1} )),
    \]
    as required.
\end{proof}

The next lemma will be useful in the proof of 
Lemma~\ref{lem: one crossing bp} 
\begin{lemma}
    \label{app: bp order}
    For any \(S \subseteq N\) and \(1 \leq x\leq p\), \(b(S, x) \geq b(S, x+1).\) 
\end{lemma}

\begin{proof}
   Using Lemma \ref{lem: balance point}, we have that 
    \begin{align}
b(S, x+1) 
&= \frac{1}{2} \left(1 + \alpha(S^x) - \frac{q(S^x)}{q^{x+1}} \alpha(S^x) \right) \nonumber\\
&= \frac{1}{2} \left(1 + \alpha(S^{x-1}) + \alpha(S \cap Q_x) 
       - \frac{q(S^x) \alpha(S^x)}{q^{x+1}} \right) \label{eq:bp-order-1}\\
&= \frac{1}{2} \left(1 + \alpha(S^{x-1}) + \alpha(S \cap Q_x) 
       - \frac{\sum_{1 \leq y \leq x} q^y \alpha(S \cap Q_y)}{q^{x+1}} \right) \label{eq:bp-order-2}\\
&\leq \frac{1}{2} \left(1 + \alpha(S^{x-1}) + \alpha(S \cap Q_x) 
       - \frac{\sum_{1 \leq y \leq x} q^y \alpha(S \cap Q_y)}{q^x} \right) \label{eq:bp-order-3}\\
&= \frac{1}{2} \left(1 + \alpha(S^{x-1}) 
       - \frac{\sum_{1 \leq y < x} q^y \alpha(S \cap Q_y)}{q^x} \right)\nonumber.
\end{align}
    In \eqref{eq:bp-order-1} we use the fact that 
    \[
    S^x = S \cap (\cup_{1 \leq y \leq x}Q_y) = S^{x-1} \cup (S \cap Q_x),
    \]
    and \eqref{eq:bp-order-2} holds because \(S^x = \cup_{1 \leq y \leq p} (S \cap Q_y)\), as the quality groups are distinct. For inequality \eqref{eq:bp-order-3}, \(q^x \geq q^{x+1}\) due to sorting.

    Applying Lemma \ref{lem: balance point}, we have
    \begin{align*}
b(S, x+1) 
&\leq \frac{1}{2} \left(1 + \alpha(S^{x-1}) - \frac{\sum_{1 \leq y < x} q^y \alpha(S \cap Q_y)}{q^x} \right) \\
&= \frac{1}{2} \left(1 + \alpha(S^{x-1}) - \frac{q(S^{x-1}) \cdot \alpha(S^{x-1})}{q^x} \right) = b(S, x),
\end{align*}
which completes the proof.
\end{proof}

\begin{proof}[\bf Proof of Lemma~\ref{lem: one crossing bp}:]
    Consider any \(\bar{S}_j\) with \(j \geq i\), where \(j \in Q_y\) (with \(y \geq x\)). 
    By Lemma \ref{app: bp order}, we have that
    \[
    \alpha(\bar{S}_i \cup \{j\}) \geq \alpha(\bar{S}_i) \geq b(S,x) \geq b(S, y).
    \]
    Hence, by the definition of balance point (\ref{def: bp}), \(g(\bar{S}_i) \geq g(\bar{S}_i \cup \{j\})\).
    Iterating this argument over all subsequent agents gives
    \[
    g(\bar{S}_i) \geq g(\bar{S}_i \cup \{i+1\}) \geq g(\bar{S}_{i+1} \cup \{i+2\}) \geq \dots \geq g(S),
    \]
    as required.
\end{proof}

\begin{proof}[\bf Proof of Lemma~\ref{lem: comp bp 1}:]

Consider a quality group $Q_x$ for some $1 \leq x \leq z$.  
We define an auxiliary agent $r$ with quality $q_r = q^x$ and share $\alpha_r = \alpha(S \cap Q_x) - \alpha(S' \cap Q_x)$, which compensates for the gap between the shares taken by $S$ and $S'$ over this quality group.  
Note that this auxiliary agent is well-defined, since $\alpha_r \geq 0$.
We show that when agent $r$ is added to $S$, all subsequent balance points—i.e., $b(S, y)$ for $y > x$—decrease, while the preceding ones remain unchanged.  
Thus, by repeatedly applying this argument across all quality groups, we construct a hypothetical hired group of agents equivalent to $S'$, for which none of the balance points increase, as required.

Let us now analyze the effect of adding agent $r$ on each balance point up to $z$.  
For any quality group $Q_y$ with $1 \leq y \leq x$, the balance point $b(S, y)$ remains unchanged, since agents of lower-quality groups do not influence the balance points of higher-quality groups, by Definition~\ref{def: bp}.  
Furthermore, for every quality group $Q_y$ with $x < y$, using Lemma~\ref{lem: balance point} and standard algebraic manipulations, we have
    \begin{align*}
    b(S \cup \{r\}, y) 
    &= \frac{1}{2}\left(1 + \alpha(S^{y - 1} \cup \{r\})\left(1 - \frac{q(S^{y-1} \cup \{r\})}{q^y}\right)\right) \\
    &= \frac{1}{2} \left(1 + \alpha(S^{y - 1} \cup \{r\}) 
         - \frac{\alpha_r q_r + \alpha(S^{y-1}) q(S^{y-1})}{q^y} \right) \\
    &= \frac{1}{2} \left(1 + \alpha(S^{y - 1}) 
         - \frac{\alpha(S^{y-1}) q(S^{y-1})}{q^y} 
         + \alpha_r \left(1 - \frac{q_r}{q^y} \right) \right) \\
    &= b(S, y) + \frac{1}{2} \alpha_u \left(1 - \frac{q_r}{q^y} \right).
    \end{align*}
    Since \(x > y\) implies \(q_r = q^x > q^y\), we have that \(\frac{1}{2}\alpha_r(1 - \frac{q_r}{q^y}) \leq 0\). 
    Therefore,
    \[
    b(S \cup \{r\}, y) \leq b(S, y),
    \]
    as required.
\end{proof}

\begin{proof}[\bf Proof of Lemma~\ref{lem: comp bp 2}:]
Given the sets $S$ and $S'$, we construct a group of hired agents $T$ such that 
\[
g(S) \leq g(T) \leq g(S').
\]
This construction uses an auxiliary agent, so $T$ is not necessarily a subset of $N$.

We initialize $T \leftarrow \emptyset$.  
Then, we iterate over each agent $i \in \bar{S}'$ in order, adding $i$ to $T$ as long as $\alpha(T \cup \{i\}) < \alpha(S)$. Once this inequality no longer holds, we stop the process.  
Let $t \in S'$ denote the agent that caused the process to stop.
    Since
    \[
    \alpha(S') = \sum_{1 \leq x \leq p} \alpha(S' \cap Q_x) \geq \sum_{1 \leq x \leq p} \alpha(S \cap Q_x) = \alpha(S),
    \]
    such an agent \(t\) must exist. 
    Let $Q_z$ be its quality group.

    To finalize \(T\), we add an auxiliary agent \(t'\) with \(q_{t'} = q_t\) and \(\alpha_{t'} = \alpha(S) - \alpha(T)\), leading to \(\alpha(T) = \alpha(S)\).

By the quality group relation between $S$ and $S'$, the set $T$ reallocates some of the shares previously assigned to $Q_z$ to higher-quality agents, thereby improving the overall quality of the group. Therefore,
    \[
    g(S) = (1 - \alpha(S)) \alpha(S) q(S) = (1 - \alpha(T)) \alpha(T) q(S) \leq (1 - \alpha(T)) \alpha(T) q(T) = g(T).
    \]

Furthermore, since $T^{z-1} = S'^{z-1}$ and $\alpha(T^z) \leq \alpha(S'^z) \leq b(S', z) = b(T, z)$, it follows from the properties of balance points that
\[
g(S'^z) \geq g(T^z) = g(T).
\]
We note that adding agents from lower-quality groups $Q_x$ with $x > z$ only increases utility, as long as the allocations remain within their respective balance points.  
Therefore, we have that
\[
g(S) \leq g(T) \leq g(S').
\]
\end{proof}

\begin{proof}[\bf Proof of Lemma~\ref{lem: comp bp 3}:]
    We extend the construction in the proof of Lemma~\ref{lem: comp bp 2}. 
    We build \(T\) over \(S^y\) and \(S'^y\) such that \(g(S) \leq g(T) \leq g(S'^y \cup \{r\})\).

    As before, we initialize \(T \leftarrow \emptyset\) and iterate over agents \(i \in \bar{S}'^y\) in order, adding them to \(T\) as long as \(\alpha(T \cup \{i\}) < \alpha(S^y)\). 
    We stop at agent \(t \in S'^y\), which must exist since
    \[
    \alpha(S'^y) = \sum_{1 \leq x \leq y} \alpha(S' \cap Q_x) \geq \alpha(S^y).
    \]
    Let \(t \in Q_z\) be this agent's quality group.

    We now add two auxiliary agents: 
    \(t'\) with \(q_{t'} = q^z\) and \(\alpha_{t'} = \alpha(S^y) - \alpha(T)\), and \(t''\) with \(q_{t''} = q^{y+1}\) and \(\alpha_{t''} = \alpha(S) - \alpha(T)\). 
This ensures that $\alpha(T) = \alpha(S)$ and $q(T) \geq q(S)$,  
since the shares originally assigned in $S$ to agents of lower quality than $q^y$  
(i.e., those with quality at most $q^{y+1}$) are effectively replaced by shares allocated to agents of quality $q^{y+1}$.

Finally, since $T^y = S'^y$ by construction and $g(T^{y+1}) \leq g(S'^y \cup \{r\})$ by Observation~\ref{obs: aux agent bp}, we obtain:
\[
g(S) \leq g(T) = g(T^{y+1}) \leq g(S'^y \cup \{r\}),
\]
completing the proof.
\end{proof}

\section{From Online Budgeted Maximization To OMAC}
\label{app:ks to omac}

The technique developed in Sections~\ref{sec:balance-points} and~\ref{sec:rand-alg-additive} centers on the notion of \emph{balance points}—dynamic thresholds that capture the optimal tradeoff between allocating shares to agents and preserving the principal’s utility. To underscore the significance of this technique for our positive result, we present a simpler but weaker alternative, which treats the threshold statically. Rather than adjusting balance points dynamically as the process evolves, we fix a global constraint on the total share that may be distributed.

Viewing the problem through this lens establishes a structural alignment between OMAC and the broader landscape of online budgeted maximization with preemption. Namely, the problem becomes closely related to the classical \emph{Online Knapsack with preemption} (OKS) problem~\cite{RemovableKS}.
In OKS, we are given a budget $\beta \in \mathbb{R}_{>0}$ and a sequence of items arriving one by one. Each item $e \in E$ is characterized by a value $p_e$ and a cost $c_e$. Upon the arrival of each item, the algorithm must decide whether to accept it, with the option to subsequently discard previously accepted items. The goal is to maximize the total value $p(V) = \sum_{e \in V} p_e$ of the accepted items, subject to the constraint $\sum_{e \in V} c_e \le \beta$.

We can obtain a direct reduction from OKS to OMAC. Consider a sequence of OKS items $E$, each with value $p_e$ and cost $c_e$. For every budget $\beta > 0$, define an order-preserving mapping $\Phi_\beta : E \to N$, such that if item $e$ arrives before item $v$, then $\Phi_\beta(e)$ arrives before $\Phi_\beta(v)$. Under this mapping, each item $e$ becomes an agent $i = \Phi_\beta(e)$ with reward $f(\{i\}) = p_e$ and share $\alpha_i = c_e$. Extend $f$ additively over sets, so that $f(S) = \sum_{i \in S} f(\{i\})$.

With this setup, any subset $V \subseteq E$ of selected items corresponds to a set of agents $S = \Phi_\beta(V)$ such that $f(S) = \sum_{e \in V} p_e$ and $\alpha(S) = \sum_{e \in V} c_e \le \beta$. Thus, selecting items in OKS is equivalent to hiring agents in OMAC, and the option of free dismissal is preserved under the mapping $\Phi_\beta$.

Formally, this correspondence implies that the $1/2$-competitive algorithm for OKS from~\cite{RemovableKS} yields a $1/4$-competitive algorithm for OMAC under the reduction described above, when using a fixed budget $\beta = 1/2$. We denote this algorithm by ${\sf ALG}_{\text{OKS}}$, defined by running \textsc{Greedy} (equivalent to Algorithm~\ref{algo: bp} with a constant threshold $\beta$ instead of balance points) and \textsc{Max} (which selects the item with the highest profit, as opposed to Algorithm~\ref{alg: max}, which chooses the agent maximizing the principal’s utility), each with probability $1/2$.

Given a budget $\beta$, we use ${\sf ALG}_{\text{OKS}}$ together with the reduction $\Phi_\beta$ to construct a corresponding algorithm for OMAC, denoted ${\sf ALG}_{\text{OMAC}}^\beta$. For any sequence of items $V$, this algorithm returns the hired agent set $\Phi_\beta({\sf ALG}_{\text{OKS}}(V))$.

To handle edge cases in our setting, we extend \textsc{Max} to also consider elements whose cost exceeds the budget (but remains at most $1$), to avoid degenerate behavior. For example, without this extension, an item $e \in E$ with cost $c_e \in (\beta, 1]$ would be ignored, and if the input consists of only this item (e.g., with $c_e = \beta + \varepsilon$ for arbitrarily small $\varepsilon > 0$), ${\sf ALG}_{\text{OMAC}}^\beta$ would fail to select it—yielding a competitive ratio of $0$.

This construction leads to the following tight guarantee, indicating that the static-threshold approach fails to capture the full complexity of OMAC:

\begin{theorem}
    \label{theo:ks and omac pos}
    Let ${\sf ALG}_{\text{OMAC}}^\beta$ be as defined above. Then:
    \begin{itemize}
        \item The algorithm ${\sf ALG}_{\text{OMAC}}^{1/2}$ is $1/4$-competitive for OMAC.
        \item For any fixed budget $\beta > 0$, the competitive ratio of ${\sf ALG}_{\text{OMAC}}^\beta$ is at most $1/4$.
    \end{itemize}
\end{theorem}

While this approach is weaker than the balance-point framework (which achieves a competitive ratio of $1/2$), it reveals a structural similarity between the two problems and underscores why dynamic thresholds are crucial for attaining tight guarantees in OMAC. Additionally, it highlights the potential value in studying the OKS model in its own right. The remainder of this section is dedicated to proving the above results.

Let $E$ be a set of elements, and let $\OPT_{\text{OKS}} \subseteq E$ and $\OPT_{\text{OMAC}} \subseteq N = \Phi_\beta(E)$ denote optimal solutions for the OKS instance and its corresponding OMAC instance, respectively.

\begin{lemma}
    \label{lem:ks and omac pos 1}
    The expected profit of the algorithm in OKS provides a 2-approximation to the principal’s reward in OMAC; that is,
    \[
    2\,\mathbf{E}\big[p({\sf ALG}_{\text{OKS}}(E))\big] \;\;\ge\;\; f(\OPT_{\text{OMAC}}).
    \]
\end{lemma}

\begin{proof}
We aim to prove that
\[
f(\OPT_{\text{OMAC}}) = p(\Phi^{-1}(\OPT_{\text{OMAC}})) \leq p(\textsc{Greedy}(E)) + p(\textsc{Max}(E)),
\]
from which the claim follows immediately by taking expectations.

We consider two cases. If \( \alpha(\OPT_{\text{OMAC}}) \leq \beta = \tfrac{1}{2} \), then \( \Phi^{-1}(\OPT_{\text{OMAC}}) \) is a feasible solution to OKS, and so
\[
f(\OPT_{\text{OMAC}}) \leq p(\OPT_{\text{OKS}}).
\]
It is well known that
\[
p(\OPT_{\text{OKS}}) \leq p(\textsc{Greedy}) + p(\textsc{Max}),
\]
which completes the proof in this case.

Now consider the case where \( \alpha(\OPT_{\text{OMAC}}) > \tfrac{1}{2} \). Following Equation~\eqref{eq:quality-set}, define the {\em quality} of a set \( U \subseteq E \) as \( q(U) = \frac{p(U)}{c(U)} \), and the quality of an element \( e \in E \) as \( q_e = \frac{p_e}{c_e} \). Let \( G \subseteq E \) be the output of \textsc{Greedy} on \( E \), and define
\[
v = \arg\max \{\, q_e \mid e \in E \setminus G \,\}.
\]

Note that \( \Phi^{-1}(\OPT_{\text{OMAC}}) \subseteq E \), and this set is nonempty because \textsc{Greedy} cannot select all elements in \( E \) without exceeding the budget. Therefore, some elements—including potentially \( v \)—must remain outside \( G \).

By construction, \textsc{Greedy} selects the highest-quality elements in order until the budget is exhausted. Thus, the set \( G \cup \{v\} \) contains the \( |G| + 1 \) elements of highest quality in \( E \). It follows that:
\[
q(G \cup \{v\}) \geq q(\OPT_{\text{OKS}}).
\]

Moreover, since \( v \notin G \), we know that including \( v \) must exceed the budget, implying
\[
c(G \cup \{v\}) > \tfrac{1}{2}.
\]
In particular, if we assume towards contradiction that
\[
c(G \cup \{v\}) < c(\Phi^{-1}(\OPT_{\text{OMAC}})),
\]
then, by the structure of the reduction \( \Phi \) and the definition of the principal’s utility in Equation~\eqref{eq: g}, it would follow that: \begin{align}
    g(\Phi(G \cup \{v\})) 
    &= (1 - \alpha(\Phi(G\cup\{u\})))\alpha(\Phi(G\cup\{u\})) q(\Phi(G\cup\{u\}))) \nonumber\\
    &= (1 - c(G\cup\{u\}))c(G\cup\{u\}) q(G\cup\{u\})) \nonumber\\
    &\geq (1 - c(G\cup\{u\}))c(G\cup\{u\}) q(\OPT_{OMAC})  \nonumber\\
    &> (1 - c(\Phi^{-1}(\OPT_{OMAC})))c(\Phi^{-1}(\OPT_{OMAC})) q(\OPT_{OMAC}) \label{eq:ks-algo-profit-1}\\
    &= (1 - \alpha(\OPT_{OMAC}))\alpha(\OPT_{OMAC}) q(\OPT_{OMAC})\nonumber\\
    &= g(\OPT_{OMAC})\nonumber,
\end{align}
    where the second inequality \eqref{eq:ks-algo-profit-1} is since $c(G \cup \{v\}) > c(\OPT_{OKS}) \geq 1/2$, leading to a contradiction to optimality.

    To conclude, since $q(G \cup \{v\}) \geq q(\OPT_{OMAC})$ and $c(G \cup \{v\}) \geq c(\Phi^{-1}(\OPT_{OMAC}))$, we obtain that $$p(G \cup \{v\}) = c(G \cup \{v\})q(G \cup \{v\}) \geq c(\Phi^{-1}(\OPT_{OMAC}))q(\Phi^{-1}(\OPT_{OMAC}) = p(\Phi^{-1}(\OPT_{OMAC})),$$
    and so $p(\Phi^{-1}(\OPT_{OMAC})) \leq p(G) + p_u \leq p(G) + p($\textsc{Max}$(E))$, as required.
\end{proof}

\begin{proof}[Proof of Theorem \ref{theo:ks and omac pos}]
    For the output corresponding to OMAC, $S = \Phi({\sf ALG}_{OKS}(E))$, we get that: \begin{align*}
        g(\OPT_{OMAC}) = 
        &(1 - \alpha(\OPT_{OMAC}))f(\OPT_{OMAC}) \leq (1 - \alpha(\OPT_{OMAC}))2\mathbf{E}[f(S)] \\
        &\leq 2 \mathbf{E}[f(S)] = 2 \mathbf{E}\big[\frac{g(S)}{1 - \alpha(S)}\big] \leq 4\mathbf{E}[g(S)].
    \end{align*}
    The first and second equalities are by the definition of the principal utility \eqref{eq: g}, the first inequality is by Lemma \ref{lem:ks and omac pos 1} and the third inequality is the budget imposed by the reduction ($\mathbf{E}[\alpha(S)] = \mathbf{E}[c(V)] \leq \beta = \frac{1}{2}$).

For the lower bound, we analyze two cases based on the value of the budget parameter $\beta$.

\textbf{Case 1: $\beta < \frac{1}{2}$}.  
We construct the following instance. Let $\varepsilon \in \mathbb{R}_{> 0}$, and define the element set \( E = \{v, u_1, u_2\} \) with:
\begin{itemize}
    \item Profits: \( p(u_1) = p(u_2) = 1 \), \( p(v) = 1 + \varepsilon \),
    \item Costs: \( c(u_1) = c(u_2) = \beta + \varepsilon \), \( c(v) = 2(\beta + \varepsilon) \).
\end{itemize}
The arrival order is irrelevant.

Since no element fits within the budget individually, \textsc{Greedy} selects nothing. Meanwhile, \textsc{Max} selects item \( v \), as it has the highest profit and \( c(v) < 1 \) (given \( \beta < \frac{1}{2} \)). This yields an expected utility of:
\[
\frac{1}{2}(1 - c(v)) \cdot p(v) = \frac{1}{2}(1 - 2(\beta + \varepsilon))(1 + \varepsilon).
\]
On the other hand, a solution selecting both \( u_1 \) and \( u_2 \) yields:
\[
(1 - c(u_1) - c(u_2)) \cdot (1 + 1) = (1 - 2(\beta + \varepsilon)) \cdot 2.
\]
Hence, the competitive ratio in this case is strictly less than \( \frac{1}{4} + O(\varepsilon) \), as desired.

\textbf{Case 2: $\beta \geq \frac{1}{2}$}. 
We construct another instance with the following setup:
\[
E = \{v\} \cup E_1,
\]
where:
\begin{itemize}
    \item \( p_v = \varepsilon \), \( c_v = \beta - \varepsilon \),
    \item For every \( e \in E_1 \): \( p_e = \varepsilon \), \( c_e = \varepsilon^2 \),
    \item \( |E_1| = \frac{1}{\varepsilon} \).
\end{itemize}

Observe that:
\[
c(E) = c_v + \sum_{e \in E_1} c_e = \beta - \varepsilon + \frac{1}{\varepsilon} \cdot \varepsilon^2 = \beta,
\]
so the total cost matches the budget, and \textsc{Greedy} can select all elements. The total profit is:
\[
p(E) = p_v + \sum_{e \in E_1} p_e = \varepsilon + \frac{1}{\varepsilon} \cdot \varepsilon = 1 + \varepsilon.
\]
Hence, the utility of \textsc{Greedy} is:
\[
g(\Phi(E)) = (1 - c(E)) \cdot p(E) = (1 - \beta)(1 + \varepsilon).
\]
Meanwhile, since all items have identical profit, \textsc{Max} selects one arbitrarily, achieving utility at most \( \varepsilon \). Therefore, the expected utility of the algorithm is:
\[
\frac{1}{2} \cdot (1 - \beta)(1 + \varepsilon) + \frac{1}{2} \cdot \varepsilon = \frac{1 - \beta}{2}(1 + \varepsilon) + \frac{\varepsilon}{2}.
\]
Since \( \beta \geq \frac{1}{2} \), we obtain:
\[
\text{Expected utility} \leq \frac{1}{4} + O(\varepsilon).
\]

In contrast, selecting only \( \Phi(E_1) \) yields utility:
\[
g(\Phi(E_1)) = (1 - c(E_1)) \cdot p(E_1) = (1 - \varepsilon) \cdot 1 = 1 - \varepsilon.
\]
Thus, we have shown that the competitive ratio is at most \( \frac{1}{4} + O(\varepsilon) \) in all cases.
\end{proof}

\section{Appendix for Section~\ref{sec:negative}}
\label{app:neg-add}

\subsection{Additive Rewards} 
\begin{proof}[Proof of Theorem \ref{theo:add-neg-det}]
Let $\varepsilon \in (0, 1)$ be a small positive value, and set $n = \frac{1}{2\varepsilon^2} + 1$. Consider the following sequence of agents: $N = \{1, \dots, n\}$, listed in their order of arrival, such that $\alpha_1 = 1 - \varepsilon^2$, $q_1 = \varepsilon$, and for every $2 \leq i \leq n$, we have $\alpha_i = \varepsilon^2$ and $q_i = \varepsilon^2$. 
We refer to all agents in $N \setminus \{1\}$ as \emph{small agents}.

Any deterministic algorithm can be classified into one of the following categories w.r.t $N$:
    \begin{enumerate}
    \item ${\sf ALG}_1$, a class of algorithms that reject agent $1$ upon arrival.
    
    \item $\forall i \in N \setminus \{1\}: {\sf ALG}_i^-$. An algorithm in this class initially hires agent $1$ and dismisses it following the arrival of agent $i$ (if such an arrival occurs). Furthermore, at the moment immediately after agent $i$’s arrival—but prior to making any further hiring decisions—algorithms in this class have no hired small agents.
    
    \item For all $i \in N \setminus \{1\}$, let ${\sf ALG}_i^+$ denote the class of algorithms that hire agent $1$ and dismiss them upon the arrival of agent $i$ (if such an arrival occurs). In addition, at the time of agent $i$'s arrival—before any decision is made—algorithms in this class have already hired at least one small agent.
    
    \item ${\sf ALG}_{n+1}$, a class of algorithms that hire agent $1$ and do not dismiss them.
    \end{enumerate}

    We claim that for each of the algorithm classes defined above, there exists a value $k$ such that the utility of every algorithm in the class on the prefix $N_k$ of the input sequence is arbitrarily small relative to the optimal.
    In the remainder of the proof, we focus on selecting an appropriate value of $k$ that ensures the claim holds uniformly across all algorithm classes, as detailed below.

    For the class ${\sf ALG}_1$, we choose $k = 1$. At this point, the input sequence has stopped, and every algorithm in this class has achieved a utility of $0$, while the optimal solution would obtain a strictly positive utility. Specifically, if the principal hires the first agent, the resulting utility is
$$
g(\{1\}) = (1 - \alpha_1) \cdot q_1 \cdot \alpha_1 = (1 - \varepsilon^2) \cdot \varepsilon^2 \cdot \varepsilon > 0.
$$

For the class ${\sf ALG}_i^-$, we choose $k = i$. By definition, upon arrival of agent $i$, an algorithm in ${\sf ALG}_i^-$ holds no small agents and, in particular, dismisses agent $1$ at this point. Consequently, its utility is at most $g(\{i\})$, since such an algorithm either hires agent $i$ or no agent at all. Observe that
$$
g(\{i\}) = (1 - \alpha_i) \cdot \varepsilon^2 \cdot \alpha_i = (1 - \varepsilon^2) \cdot \varepsilon^2 \cdot \varepsilon^2 = \varepsilon \cdot g(\{1\}).
$$
Therefore, as $\varepsilon \to 0$, we have $CR \to 0$.

For the class ${\sf ALG}_i^+$, we set $k = i - 1$. Since the input sequence $N_k$ is a prefix of the full arrival sequence, it follows that an algorithm from ${\sf ALG}_i^+$ has not yet dismissed agent $1$—as the input ends just before it would make that decision. Moreover, the algorithm’s selected group of agents includes agent $1$ and at least one of the small agents that have arrived thus far. However, hiring more than one small agent alongside agent $1$ results in a non-positive utility, as the total share allocated would be at least $1$. In contrast, the optimal solution can achieve a strictly positive utility in many different ways. This implies that the competitive ratio satisfies $CR \leq 0$.

Finally, for the class ${\sf ALG}_{n+1}$, we set $k = n$, so that algorithms in this class receive the full input sequence $N_n = N$. As discussed earlier, in scenarios where the algorithm never dismisses agent $1$, its utility is at most $g(\{1\})$, since hiring additional small agents after hiring agent $1$ does not improve the utility. Nevertheless, the optimal solution could instead hire all $n - 1$ small agents, forming the group $N \setminus \{1\}$, and thereby achieve a utility 
 \begin{align*}
    g(N \setminus \{1\})
&=(1 - \alpha(N \setminus \{1\})) \alpha(N \setminus \{1\}) q(N \setminus \{1\})\\ 
&= (1 - (n-1) \cdot \alpha_i) \cdot (n-1) \cdot \alpha_i \cdot \varepsilon^2 \\
&= \left(1 - \frac{1}{2\varepsilon^2} \cdot \varepsilon^2 \right) \cdot \frac{1}{2\varepsilon^2} \cdot \varepsilon^2 \cdot \varepsilon^2 \\
&= \frac{\varepsilon^2}{4} 
= \frac{\varepsilon^3 (1 - \varepsilon^2)}{4\varepsilon (1 - \varepsilon^2)}
= \frac{g(\{1\})}{4\varepsilon (1 - \varepsilon^2)},
\end{align*}
    resulting in arbitrarily small CR. This completes the proof.
\end{proof}

\begin{proof}[Proof of Theorem~\ref{thm:CR_rand_upper_bound}]
To complete the proof from the main body, we show the next claim. 
\begin{claim}
For any $1 \le i \le n+1$, letting $k = \min \{i, n \}$, the utility achieved by any algorithm in the class ${\sf ALG}_i$ on the input prefix $N_k$ can be made arbitrarily small relative to the optimal solution.
\end{claim}
\begin{proof}
For ${\sf ALG}_1$, we choose $k = 1$. In this case, the optimal solution is to hire agent $1$. An algorithm achieves a competitive ratio (CR) of $1$ when event $X$ occurs, and $0$ otherwise (since it obtains zero utility on $\overline{X}$), resulting in an expected CR
bounded above by $\Pr(X) \leq 1/2$, as required.

For ${\sf ALG}_i$, where $2 \leq i \leq n$, we set $k = i$. Here, the utility obtained when event $X$ occurs is negligible. To see this, we consider two cases depending on whether the algorithm keeps agent $1$ or not, and show that in either case the expected utility is negligible compared to that of hiring agent $1$ alone.

If an algorithm retains agent $1$, then by the definition of the principal's utility \eqref{eq: g}, we have that 
\begin{align} 
&\mathbf{E}[g(T_k \cup \{1\}) \mid X] = \nonumber\\
&\mathbf{E}[(1-\alpha(T_k \cup \{1\}))(\alpha(T_k)q(T_k) + \alpha_1 q_1)] \leq \nonumber\\
&(1 - \mathbf{E}[\alpha(T_k)] - \alpha_1)(\mathbf{E}[\alpha(T_k)] q_i + \alpha_1 q_1) \label{eq:add-neg-det-1} = \\
&(1 - \mathbf{E}[t_k] \alpha_i - \alpha_1)(\mathbf{E}[t_k] \alpha_i q_i + \alpha_1 q_1) \label{eq:add-neg-det-2} = \\
&(\varepsilon^2 - \mathbf{E}[t_k] \varepsilon^2)(\mathbf{E}[t_k] \varepsilon^4 + \varepsilon - \varepsilon^3) \leq \label{eq:add-neg-det-3}\\
&\varepsilon^3 (1 - \mathbf{E}[t_k]) + \varepsilon^6 (1 - \mathbf{E}[t_k]) \mathbf{E}[t_k] \leq \nonumber \\
&\varepsilon^3 (1 - \mathbf{E}[t_k]) + O(\varepsilon^6) = \label{eq:add-neg-det-4}\\
&\varepsilon^3 (1 - \mathbf{E}[t_k]) + O(\varepsilon^3) \cdot g(\{1\}) \label{eq:add-neg-det-5},
\end{align}

We note that \eqref{eq:add-neg-det-1} follows from Jensen's inequality and the concavity of $g$ in its continuous form over $\alpha$.  
Equality \eqref{eq:add-neg-det-2} uses the definitions of small agents and the set $T_k$.
The inequality in \eqref{eq:add-neg-det-4} holds since $(1 - \mathbf{E}[t_k]) \mathbf{E}[t_k] \leq 1$, and the final equality \eqref{eq:add-neg-det-5} substitutes the value of $$g(\{1\})  = (1 - \alpha_1) \cdot q_1 \cdot \alpha_1 = (1 - \varepsilon^2) \cdot \varepsilon^3.$$
Together, these steps imply that 
\[
\mathbf{E}[g(T_k \cup \{1\}) \mid X] < O(\varepsilon^3) \cdot g(\{1\}),
\]
showing that the utility in this case is negligible.

If, instead, an algorithm in this class does not keep agent $1$, then by $\mathbf{E}[t_{k-1}] \leq 1$ and the definition of $t_k$, it follows that $\mathbf{E}[t_k] \leq 2$. This is because the algorithm can increase $t_k$ by at most one upon accepting the arriving agent with probability $1$, i.e., $T_k = T_{k-1} \cup \{k\}$. 
The utility obtained from hiring two small agents is insignificant compared to that of agent $1$ alone. Consequently, the utility $\mathbf{E}[g(T_k)]$ also approaches $0$ as $\varepsilon \to 0$.

To summarize the case of the class ${\sf ALG}_i$, note that its expected competitive ratio satisfies  
\[
\begin{array}{ll}
\mathbf{E}[\textit{CR of } {\sf ALG}_i] 
&= \mathbf{E}[\textit{CR of } {\sf ALG}_i \mid X] \cdot Pr(X) \\ \\
&+ \mathbf{E}[\textit{CR of } {\sf ALG}_i \mid \bar{X}] \cdot (1 - Pr(X)) \\ \\
&\leq O(\varepsilon) \cdot Pr(X) + 1 - Pr(X) < 1/2 + O(\varepsilon),
\end{array}
\]
    as required.

For ${\sf ALG}_{n+1}$, we set $k = n$.  
We again show that the competitive ratio obtained when event $X$ occurs is $O(\varepsilon)$.
If the algorithm did not keep agent $1$, its utility is at most that of a single hired small agent, meaning $\mathbf{E}[g(T_k)]$ is negligible. Otherwise, we get that:
\[
\begin{array}{ll}
\mathbf{E}[g(T_k \cup \{1\}) \mid X] 
&\leq \mathbf{E}[g(T_k)] + g(\{1\}) \\ \\
&\leq (1 - \mathbf{E}[t_k] \alpha_i) \cdot \mathbf{E}[t_k] \alpha_i q_i + g(\{1\}) \\ \\
&\leq \alpha_i q_i + g(\{1\}) = \varepsilon^4 + g(\{1\}) \leq \varepsilon^3 + \varepsilon^4.
\end{array}
\]
However, if the algorithm were instead to hire all of the $n$ small agents to arrive, it would obtain a utility of $$g(N \setminus \{1\}) = (1 - (n-1) \cdot \varepsilon^2)(n-1) \cdot \varepsilon^2 \cdot \varepsilon^2 = \varepsilon^2 / 4,$$ resulting in a CR that approaches $0$ when $\varepsilon \rightarrow 0$.
\end{proof}
This completes the proof of the theorem.
\end{proof}

\subsection{XOS Rewards} 
\label{app: xos proof example}
\label{app: rem SM neg 1}

\begin{proof}[\bf Proof of Lemma~\ref{lem: aux rem SM neg}:]
    We prove by induction over $0 \leq k < m$ that for every $1 \leq h \leq m$, it holds that $A^{m-k}_h = \frac{hn + k}{n m}$. Thus, $$A^1_1 = A^{m - (m-1)}_1 = \frac{1\cdot n + m - 1}{nm} \leq \frac{n + m}{nm},$$ as required.

    \textbf{Base:} $k  = 0$. For every $1 \leq h \leq m$ it immediately holds that $A_s^{m-0} = \frac{h}{m} = \frac{hn + 0}{nm}$.

    \textbf{Step}: assume the claim holds for $A_h^{m-k}$ for every $1 \leq h \leq m$. It follows that: \begin{align*}
A^{m - (k + 1)}_h 
&= \frac{1}{n}A^{m-k}_{h+1} + \frac{n-1}{n}A^{m-k}_h \\
&= \frac{1}{n} \cdot \frac{(h+1)n + k}{nm} + \frac{n-1}{n} \cdot \frac{hn + k}{nm} \\
&= \frac{(h+1)n + k}{n^2 m} + \frac{(n-1)(hn + k)}{n^2 m} \\
&= \frac{nh + n + k + n^2 h - nh + kn - k}{n^2 m} \\
&= \frac{n^2 h + kn + n}{n^2 m} \\
&= \frac{n h + (k + 1)}{n m}.
\end{align*}
\end{proof}

\endgroup

\end{document}